\documentclass[11pt,letterpaper]{article}
\usepackage{mehdis}

\usepackage{pgfplots}
\usepackage{caption}
\usepackage{subcaption}
\usepackage{multicol}

\newcommand{\distance}{\operatorname{Dist}}
\newcommand{\overlap}{\operatorname{Overlap}}
\newcommand{\pt}{\operatorname{PT}}
\newcommand{\piprod}{\Pi_{\mathrm{Prod}}}
\newcommand{\piswap}{\Pi_{\mathrm{SWAP}}}

\newcommand{\E}{\mathop{\bf E\/}}

\newif\ifnotes\notestrue
\ifnotes
\usepackage{color}
\definecolor{mygrey}{gray}{0.50}
\newcommand{\notename}[2]{{\textcolor{mygrey}{\footnotesize{\bf (#1:} {#2}{\bf ) }}}}

\newcommand{\pnote}[1]{{\endnote{#1}}}

\else

\newcommand{\notename}[2]{{}}

\newcommand{\pnote}[1]{}

\fi

\begin{document}

\title{Testing matrix product states}

\author{Mehdi Soleimanifar \thanks{Center for Theoretical Physics, MIT. \href{mailto:mehdis@mit.edu}{mehdis@mit.edu}}\and John Wright \thanks{Department of Electrical Engineering and Computer Sciences, University of California, Berkeley. Most of this work was done while a postdoc at the Center for Theoretical Physics, MIT and a professor in the Department of Computer Science, University of Texas at Austin.  \href{mailto:wright@cs.utexas.edu}{wright@cs.utexas.edu}}\vspace{5mm} }

\date{\today}
\maketitle
\begin{abstract}

Matrix product states (MPS)
are a class of physically-relevant quantum states
which arise in the study of quantum many-body systems.
A quantum state $\ket{\psi_{1, \ldots, n}} \in \mathbb{C}^{d_1} \otimes \cdots \otimes \mathbb{C}^{d_n}$
comprised of $n$ qudits is said to be an MPS of bond dimension~$r$
if the reduced density matrix $\psi_{1, \ldots, k}$ has rank~$r$ for each $k \in \{1, \ldots, n\}$.
When $r = 1$, this corresponds to the set of product states,
i.e.\ states of the form $\ket{\psi_1} \otimes \cdots \otimes \ket{\psi_n}$,
which possess no entanglement.
For larger values of~$r$,
this yields a more expressive class of quantum states,
which are allowed to possess limited amounts of entanglement. 

Devising schemes for testing the amount of entanglement in quantum systems has played a crucial role in quantum computing and information theory.
In this work, we study the problem of testing whether an unknown state~$\ket{\psi}$
is an MPS in the property testing model.
In this model, one is given $m$ identical copies of $\ket{\psi}$,
and the goal is to determine whether $\ket{\psi}$ is an MPS of bond dimension~$r$
or whether $\ket{\psi}$ is far from all such states.
For the case of product states,
we study the product test,
a simple two-copy test previously analyzed by Harrow and Montanaro~\cite{Harrow2013_product_testing},
and a key ingredient in their proof that $\QMA(2) = \QMA(k)$ for $k \geq 2$.
We give a new and simpler analysis of the product test which achieves an optimal bound for a wide range of parameters,
answering open problems in~\cite{Harrow2013_product_testing} and \cite{montanaro2013survey}.
For the case of $r \geq 2$,
we give an efficient algorithm for testing whether $\ket{\psi}$ is an MPS of bond dimension~$r$
using $m = O(n r^2)$ copies, independent of the dimensions of the qudits,
and we show that  $\Omega(n^{1/2})$ copies are necessary for this task.
This lower bound shows that a dependence on the number of qudits~$n$
is necessary, in sharp contrast to the case of product states where a constant number of copies suffices.

 \end{abstract}
\newpage
\section{Introduction}

This paper is about matrix product states (MPS).
\begin{definition}[Matrix product states]
A quantum state $\ket{\psi} \in \mathbb{C}^{d_1} \otimes \cdots \otimes \mathbb{C}^{d_n}$
consisting of $n$ qudits
is a \emph{matrix product state with bond dimension~$r$}
if it can be written as
\begin{equation*}
\ket{\psi_{1, \ldots, n}} 
= \sum_{i_1 \in [d_1], \ldots, i_n \in [d_n]} \Tr[A^{(1)}_{i_1} \cdots A^{(n)}_{i_n}] \cdot \ket{i_1 \cdots i_n},
\end{equation*}
where each matrix $A^{(i)}_j$ is an $r \times r$ complex matrix, for $i \in [n]$ and $j \in [d_i]$.
We write $\mps_n(r)$ for the set of such states, or more simply $\mps(r)$ when the dependency on $n$ is clear from the context. 
\end{definition}
The parameter~$r$ controls the amount of entanglement $\ket{\psi}$ is allowed to possess,
and as it increases, the set of MPS grows larger and more expressive.
On one extreme, when $r = 1$ this corresponds to the set of product states, i.e.\ state of the form $\ket{\psi_1} \otimes \cdots \otimes \ket{\psi_n}$,
which possess no entanglement between different qudits. 
On the other extreme, every state $\ket{\psi}$, even a highly entangled one, is an MPS of bond dimension $r = d_1 \cdots d_n$.
Between these two extremes,
MPS allow for nonzero though still limited entanglement,
which grows with~$r$.
This can be seen more readily in the following alternative characterization of MPS,
which states that $\ket{\psi_{1, \ldots, n}}$ is an MPS of bond dimension~$r$
if and only if $\psi_{1, \ldots, k}$ has rank $r$ for each $1 \leq k \leq n$,
where $\psi_{1, \ldots, k}$ is the reduced density matrix on the first~$k$ qudits.
Here, we say that a Hermitian matrix has rank~$r$ if it has at most $r$ nonzero eigenvalues.
This implies, for example,
that the entanglement entropy between the first~$k$ and the last $n-k$ qudits is always at most $\log(r)$, for each $k$.
We will prefer this alternative characterization in this paper.

MPS feature prominently in the study of quantum many-body physics,
with a particular emphasis on one-dimensional quantum systems.
In a typical one-dimensional quantum system,
$n$ qudits are arranged on a line,
and their interactions are governed by a local Hamiltonian~$H$
which only contains local terms between neighboring qudits,
i.e. terms of the form $H_{i, i+1}$.
The \emph{one-dimensional area law} of Hastings~\cite{hastings2007area_law},
as well as further refinements in~\cite{AradFrustrationFreeAreaLaw,Arad2013AreaLaw,Landau2015polynomialTimeAlgorithm},
implies that if~$H$ is a gapped Hamiltonian,
then its ground state $\ket{\psi_{1, \ldots, n}}$,
is well-approximated by an MPS of ``small'' bond dimension.
One-dimensional quantum systems are an important class of physically-motivated systems,
and this characterization in terms of MPS means they are tractable to analyze with computers.
For example,~\cite{AradRigorousRG} have developed rigorous algorithms for approximating the ground state of a one-dimensional gapped Hamiltonian.
And~\cite{cramer2010efficientTomographyMPS} have suggested using MPS tomography to efficiently learn the state of a one-dimensional system using a small number of copies,
motivated by the fact that an MPS only has $(d_1 + \cdots + d_n) r^2$ parameters to ``learn'',
exponentially fewer than the $d_1 \cdots d_n$ parameters of a general quantum state. The classical tractability of matrix product states has also resulted in their widespread application as a computational method in the classical simulation of quantum circuits, both in one and higher dimensions. This includes the simulation of shallow quantum circuits~\cite{napp2019efficient2d,bravyi2021classicalMeanValue,Coudron2020Shallow}, slightly entangled quantum circuits~\cite{VidalClassicalSim}, and noisy quantum circuits~\cite{ZhouLimitsSimulationofQuantum}.

In this work, we study the problem of ``testing'' whether an unknown state $\ket{\psi}$ is an MPS.
We will study this in the model of property testing.
In this model, an algorithm is given access to multiple copies of $\ket{\psi}$ which it is allowed to measure;
its goal is to determine if $\ket{\psi}$ is an MPS using as few copies as possible.
This problem has been previously studied for the $r = 1$ case of product states by Harrow and Montanaro~\cite{Harrow2013_product_testing},
and studying the case of general~$r$ was suggested as an open direction by Montanaro and de Wolf~\cite{montanaro2013survey}.
To define this model, we begin by formally defining what it means for a state to be ``far'' from being an MPS.

\begin{definition}[Distance to $\mps(r)$]
Given $n\geq 1$ and a state $\ket{\psi} \in \bbC^{d_1} \ot \cdots \otimes \bbC^{d_n}$,
the distance of $\ket{\psi}$ to the set $\mps(r)$ is defined as
\begin{equation*}
\distance_r(\ket{\psi})
= \min_{\ket{\phi} \in \mps(r)}\mathrm{D}_{\mathrm{tr}}(\psi, \phi)
= \min_{\ket{\phi} \in \mps(r)}\sqrt{1 - |\braket{\psi}{\phi}|^2},
\end{equation*}
where $\mathrm{D}_{\mathrm{tr}}(\cdot, \cdot)$ denotes the standard \emph{trace distance},
and $\psi$ and $\phi$ denote the mixed states corresponding to $\ket{\psi}$ and $\ket{\phi}$, respectively.
Sometimes we will prefer to work with the maximum squared overlap of $\ket{\psi}$ with $\mps(r)$, defined as
\begin{equation*}
\overlap_r(\ket{\psi}) = \max_{\ket{\phi} \in \mps(r)} |\braket{\psi}{\phi}|^2.
\end{equation*}
When referring to the distance, we will typically use the variable name $\delta = \distance_r(\ket{\psi})$,
and when referring to the overlap, we will typically use $\omega = \overlap_r(\ket{\psi})$
or, alternatively, $1- \eps = \omega$.
Note that
\begin{equation*}
\delta
= \sqrt{1 - \omega}
=\sqrt{\eps}.
\end{equation*}
\end{definition}

Now we define the problem we consider, that of property testing MPS. 

\begin{definition}[$\mps(r)$ tester]\label{def:tester}
An algorithm $\mathcal{A}$ is a \emph{property tester for $\mps(r)$ using $m = m(n, r, \d)$ copies}
if, given $\d >0$ and $m$ copies of $\ket{\psi} \in \bbC^{d_1} \ot \cdots \otimes \bbC^{d_n}$, it acts as follows.
\begin{itemize}
\item[$\circ$](Completeness): If $\ket{\psi} \in \mps(r)$, then
		\begin{equation*}
		\Pr[\text{$\mathcal{A}$ accepts given $\ket{\psi}^{\otimes m}$}] \geq\tfrac{2}{3}.
		\end{equation*}
		If instead it accepts with probability exactly~$1$ in this case, we say that it has \emph{perfect completeness}.
\item[$\circ$] (Soundness): If $\distance_r(\ket{\psi}) \geq \d$, then 
		\begin{equation*}
		\Pr[\text{$\mathcal{A}$ accepts given $\ket{\psi}^{\otimes m}$}] \leq\tfrac{1}{3}.
		\end{equation*}
\end{itemize}
\end{definition}
\noindent
All property testers considered in this work have perfect completeness,
whereas our lower bounds will apply to property testers even with imperfect completeness.

Previous works have considered testing a variety of properties of quantum states.
Perhaps the most relevant is that of O'Donnell and Wright~\cite{Wright2015_spectrum_testing},
which considered testing properties of a mixed state $\rho$'s spectrum,
such as testing whether its rank is at most~$r$---we will revisit this later.
Another relevant work is that of Harrow, Montanaro, and Lin~\cite{HarrowSequentialMeasurements},
which considers the problem of testing whether $\ket{\psi_{1, \ldots, n}}$ is a product state across \emph{some} cut,
meaning there exists an $S \subseteq \{1, \ldots, n\}$ such that $\ket{\psi_{1, \ldots, n}} = \ket{\psi_S} \otimes \ket{\psi_{\overline{S}}}$.
If not, they say that $\ket{\psi_{1, \ldots, n}}$ possesses ``genuine $n$-partite entanglement''.
(In contrast, in $r = 1$ case of product testing, we want to verify that $\ket{\psi_{1, \ldots, n}}$ is a product state across \emph{every} cut~$S$.)
They give a tester for this problem which uses $m = O(n/\epsilon^2)$ copies of the state.
For more on quantum property testing, see the survey of Montanaro and de Wolf~\cite{montanaro2013survey}.

More broadly, testing and characterizing the entanglement of quantum systems
has been an important theme running throughout quantum computation,
even outside the model of property testing.
This includes the study of nonlocal games,
where the CHSH game~\cite{CHCH} allows one to verify that two parties share an EPR state, with applications in delegation of quantum computation~\cite{MahadevVerification,ColadangeloVerifierLeash,reichardt2013classicalLeash}, device-independent quantum cryptography~\cite{VaziraniQKD}, and interactive proof systems~\cite{MIPeqRE}.
Moreover, the communication complexity of two-party protocols for testing shared entangled states, including EPR states, has been used to reveal the properties of entanglement in ground states of local Hamiltonians~\cite{AharonovCounterExample,anshu2020spread}.  

We emphasize that we are specifically considering property testing of \emph{pure} states.
In particular, we assume that the state the algorithm $\mathcal{A}$ is given~$m$ copies of is pure, not mixed.
There are, however, problems related to ours in the property testing of mixed states, although we do not cover these in this work.
One example is the question of testing whether a \emph{mixed} state $\rho_{\mathrm{AB}}$ on two $d$-dimensional subsystems is separable (i.e.\ not entangled),
which is both fascinating and still very much open.
The best known algorithm for this problem is the trivial one:
simply use $O(d^4)$ copies to ``learn'' $\rho_{\mathrm{AB}}$ and classically compute whether it is entangled.
On the other hand, the best known lower bound is~$\Omega(d^2)$.
Another example is the problem of testing whether $\rho_{\mathrm{AB}}$ is a tensor product,
i.e.\ whether $\rho_{\mathrm{AB}} = \rho_A \otimes \rho_B$.
For this problem, we \emph{do} 
know the optimal bound: $\Theta(d^2)$ copies,
given by the algorithm of~\cite{NengkunYuIndependenceTesting}.
One convenience of pure states is that these two problems coincide for this case,
since a pure state is a product state if and only if it is unentangled.
For mixed states, this is not true.

While in this work, we focus primarily on MPS. We note that these states are a special example of the more general class of tensor network states. Devising learning and testing algorithms for these states is an interesting future direction to explore. 
\subsection{The product test}

We begin with the simplest case of MPS testing,
when the bond dimension $r = 1$,
which corresponds to testing whether $\ket{\psi} \in \mathbb{C}^{d_1} \otimes \cdots \otimes \mathbb{C}^{d_n}$ is a product state.
We study a simple two-copy property tester for this problem
known as the \emph{product test}
which was introduced by Mintert, Ku\'{s}, and Buchleitner~\cite{mintert2005}
and later studied by Harrow and Montanaro~\cite{Harrow2013_product_testing}.
The product test is itself built out of a simpler subroutine known as the \emph{SWAP test} due to Buhrman, Cleve, Watrous, and de Wolf~\cite{BuhrmanQuantumFingerprinting},
which measures the similarity between two qudit states $\ket{a}, \ket{b} \in \mathbb{C}^d$.

\begin{definition}[The SWAP test]
Given two qudit states $\ket{a}, \ket{b} \in \mathbb{C}^d$,
the SWAP test applies the two-outcome projective measurement $\{\piswap, \iden - \piswap\}$ to $\ket{a} \otimes \ket{b}$,
where $\piswap = (\iden + \mathsf{SWAP})/2$.
Here, $\mathsf{SWAP}$ is the two-qudit swap operator, defined as
\begin{equation*}
\mathsf{SWAP} \ket{i} \otimes \ket{j} = \ket{j} \otimes \ket{i}
\end{equation*}
for all $i, j \in [d]$.
The test accepts if it observes the first outcome,
and it rejects otherwise.
\end{definition}

It can be checked that the SWAP test succeeds with probability $\tfrac{1}{2} + \tfrac{1}{2} |\braket{a}{b}|^2$.
In particular, it succeeds with probability~$1$ if and only if, modulo a phase factor, $\ket{a} = \ket{b}$.
Having defined the SWAP test, we can now define the product test.

\begin{definition}[The product test]
Given two copies of a state $\ket{\psi} \in \mathbb{C}^{d_1} \otimes \cdots \otimes \mathbb{C}^{d_n}$,
the product test performs the SWAP test on the $i$-th qudit in each copy of $\ket{\psi}$, simultaneously over all $i \in [n]$,
and accepts if they all accept.
Equivalently, it performs the two-outcome projective measurement $\{\piprod, \iden-\piprod\}$,
where
$
\piprod = \piswap^{\otimes n}.
$
and the $i$-th $\piswap$ applies to the $i$-th qudits in both copies of $\ket{\psi}$.
We include an illustration of the product test in~\fig{product-test}.
\end{definition}

\begin{figure}
\centering
\begin{subfigure}[b]{0.49\textwidth}
\centering
\begin{tikzpicture}[scale=1]
\filldraw[rounded corners][fill=blue!10] (-0.6,-0.4) rectangle (4.6,0.4);
\filldraw[rounded corners][fill=blue!10] (-0.6,0.6) rectangle (4.6,1.4);

\draw[rounded corners] (-0.4,-0.6) rectangle (0.4,1.6);
\draw[rounded corners] (0.6,-0.6) rectangle (1.4,1.6);
\draw[rounded corners] (1.6,-0.6) rectangle (2.4,1.6);
\draw[rounded corners] (3.6,-0.6) rectangle (4.4,1.6);

\foreach \x in {1,...,3} {
\filldraw[fill=purple!10] (\x-1,0) node {\x} circle (0.25);
\filldraw[fill=purple!10] (\x-1,1) node {\x} circle (0.25);
 }
 \node at (3,0) {\Large ...};
 \node at (3,1) {\Large ...};
 \filldraw[fill=purple!10] (4,0) node {$n$} circle (0.25);
 \filldraw[fill=purple!10] (4,1) node {$n$} circle (0.25);
 \node[anchor=east] at (-0.8,1) {$\ket{\psi}$};
 \node[anchor=east] at (-0.8,0) {$\ket{\psi}$};
 \end{tikzpicture}
 \caption{The product test performs a SWAP test on each of the $n$ pairs of subsystems of the two copies of $\ket{\psi}.$
 Figure taken from~\cite{Harrow2013_product_testing}.}
 \label{fig:product-test}
 \end{subfigure}
\hfill
\begin{subfigure}[b]{0.49\textwidth}
\begin{tikzpicture}[scale=1]
 \filldraw[rounded corners][fill=blue!10] (-0.6+9,-1.4) rectangle (5.6+9,-0.6);
 \filldraw[rounded corners][fill=blue!10] (-0.6+9,0.6) rectangle (5.6+9,1.4);
 \draw[rounded corners] (-1 +0.6+9,-1.5) rectangle (-1 +0.6+9+0.8,1.5);
 \draw[rounded corners] (-1 +0.6+9,-1.5) rectangle (-1 +1.6+9+0.8,1.6);
 \draw[rounded corners] (-1 +0.6+9,-1.5) rectangle (-1 +2.6+9+0.8,1.7);
 \draw[rounded corners] (-1 +0.6+9,-1.5) rectangle (-1 +4.6+9+0.8,1.8);
\foreach \x in {1,...,3} {
   \filldraw[fill=purple!10] (\x-1+9,-1) node {\x} circle (0.27);
 \filldraw[fill=purple!10] (\x-1+9,1) node {\x} circle (0.27);
 }
 \node at (3+9,-1) {\Large ...};
 \node at (3+9,1) {\Large ...};
 \node at (3+9,0.1) {\Large $\vdots$};
 \node at (9,0.1) {\Large $\vdots$};
 \node at (5+9,0.1) {\Large $\vdots$};
 \node at (-1.2+9,0.1) {\Large $\vdots$};
 \filldraw[fill=purple!10] (4+9,-1) node {\small$n$-$1$} circle (0.27);
 \filldraw[fill=purple!10] (4+9,1) node {\small$n$-$1$} circle (0.27);
 \filldraw[fill=purple!10] (5+9,-1) node {$n$} circle (0.27);
 \filldraw[fill=purple!10] (5+9,1) node {$n$} circle (0.27);
 \node[anchor=east] at (-0.8+9,1) {$\ket{\psi}$};
 \node[anchor=east] at (-0.8+9,-1) {$\ket{\psi}$};
 \end{tikzpicture}
  \caption{The MPS tester simultaneously performs the rank tester on each of the $n-1$ contiguous cuts across the multiple copies of $\ket{\psi}$.}
  \label{fig:mps-tester}
 \end{subfigure}

 \caption{The product test and MPS tester.}\label{fig:product_test}
 \end{figure}
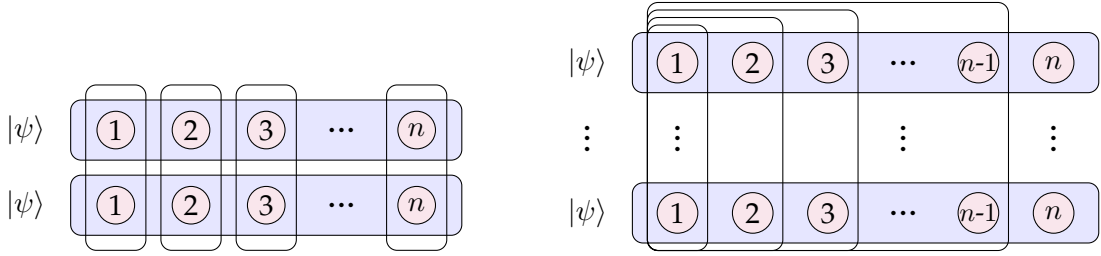
 
In the case when $\ket{\psi}$ is a product state,
i.e.\ $\ket{\psi} = \ket{\psi_1} \otimes \cdots \otimes \ket{\psi_n}$,
the product test passes with probability~$1$,
because for each $i \in [n]$
the $i$-th SWAP test is applied to $\ket{\psi_i}\otimes \ket{\psi_i}$,
and so it always succeeds.
This property of always accepting product states is known as \emph{perfect completeness}.
In fact, Harrow and Montanaro~\cite[Section 5]{Harrow2013_product_testing}
show that the product test is the optimal two-copy test for product states with perfect completeness,
in the sense that any other two-copy test with perfect completeness will reject any non-product state~$\ket{\psi}$
with at most the probability the product test rejects it.

We are interested in the maximum probability a state passes the product test, defined as follows.

\begin{definition}\label{def:PT}
 Let $n \geq 1$ and $\omega \in [0,1]$. Given a state $\ket{\psi} \in \bbC^{d_1} \ot \cdots \otimes \bbC^{d_n}$ we define $\operatorname{PT}_n(\ket{\psi})$ to be the probability the product test succeeds on $\ket{\psi}$. In addition, we define $\operatorname{PT}_n(\omega)$ to be the supremum of $\operatorname{PT}_n(\ket{\psi})$ over all $n$-partite states $\ket{\psi}$ such that $\overlap_1(\ket{\psi})= \omega$.
\end{definition}

The main result of Harrow and Montanaro~\cite{Harrow2013_product_testing}
is the following upper-bound on $\pt_n(\omega)$.
It will be more convenient to parameterize their result by $\eps$, where $1 - \eps = \omega$.

\begin{thm}[{\cite[Theorem 1]{Harrow2013_product_testing}}]\label{thm:hm}
For all $n \geq 1$ and $0 < \eps < 1$,
\begin{equation*}
\pt_n(1-\epsilon) \leq \min\{1- \epsilon + \eps^2 + \eps^{3/2}, 1 - \tfrac{11}{512} \eps\}.
\end{equation*}
Equivalently, we may write 
\begin{equation*}
\operatorname{PT}_{n}(1-\eps)\leq
\left\{\begin{array}{cl}
1- \epsilon + \eps^2 + \eps^{3/2} & \text{if $\eps \leq \eps_0$},\\
1 - \tfrac{11}{512} \eps & \text{if $\eps \geq \eps_0$,}
\end{array}
\right.
\end{equation*}
where $\eps_0 = \tfrac{1}{512} (757 - 16 \sqrt{1258}) \approx 0.37$.
We include a plot of this upper-bound in Figure~\ref{fig:hm-vs-us}.
\end{thm}

The most important regime of parameters is when $\eps$ is a constant,
in which case this result states that the product test rejects with constant probability.
This implies that two copies are sufficient to test if $\ket{\psi}$
is constantly far from being product.
\thmref{hm} is a key ingredient in Harrow and Montanaro's proof that $\QMA(2) = \QMA(k)$ for $k \geq 2$~\cite{Harrow2013_product_testing}.
Here, $\QMA(k)$ refers to \emph{Quantum Merlin Arthur with multiple certificates},
the complexity class which contains all problems solvable by a quantum polynomial-time verifier
with the help of $k$ unentangled proofs.
Their result shows that a verifier can use two unentangled copies of a proof $\ket{\psi}$
to simulate $k$ unentangled proofs
by running the product test to enforce that it is of the form $\ket{\psi_1} \otimes \cdots \otimes \ket{\psi_k}$.
As further applications of \thmref{hm},
they are able to derive hardness results for numerous (19, in fact!)
problems both in and out of quantum information theory related to entanglement, tensor optimization, and other topics.
For example, one of their applications is to the problem of detecting separability,
in which the goal is to compute whether a mixed state $\rho$ on two subsystems of dimension~$d$
(described by a $d^2 \times d^2$ complex matrix)
is separable or entangled.
They show that there exists a constant $\delta > 0$
such that if $K$ is a convex set in which every element has trace distance~$\delta$ to a separable state,
then there is no polynomial time algorithm for computing whether $\rho \in K$
unless $\text{3-SAT} \in \mathsf{DTIME}(\mathrm{exp}(\sqrt{n} \log^{O(1)}(n)))$.
See~\cite[Section 4.2]{Harrow2013_product_testing} for further details
and descriptions of the 18 other applications.

Our first result is a new and simpler analysis of the product test which yields an improved bound.
We show the following.

\begin{thm}[Product test upper-bound]\label{thm:product-test-main}
For all $n\geq 1$,
\begin{equation*}
\operatorname{PT}_{n}(\omega)\leq
\left\{\begin{array}{cl}
\omega^2 - \omega + 1 & \text{if $\omega \geq \tfrac{1}{2}$},\\
\tfrac{1}{3} \omega^2 + \tfrac{2}{3} & \text{otherwise.}
\end{array}
\right.
\end{equation*}
We include a plot of this upper-bound in Figure~\ref{fig:hm-vs-us}.
\end{thm}

To compare this with \thmref{hm},
if we set $1-\eps = \omega$ then we can rewrite this bound as
\begin{equation*}
\operatorname{PT}_{n}(1-\eps)\leq
\left\{\begin{array}{cl}
1 - \eps + \eps^2 & \text{if $\eps \leq \tfrac{1}{2}$},\\
1 - \tfrac{2}{3} \eps + \tfrac{1}{3} \eps^2 & \text{otherwise.}
\end{array}
\right.
\end{equation*}
This improves upon \thmref{hm} for all choices of $\eps > 0$, i.e.\ all $\omega < 1$,
which answers open problem no.\ 2 from~\cite{Harrow2013_product_testing}
and question no.\ 5 from~\cite{montanaro2013survey}.
In addition, the bound we achieve when $\omega \geq \tfrac{1}{2}$ is optimal, as the following well-known example shows
(cf.\ \cite[Page 31]{Harrow2013_product_testing}).

\begin{prop}[Product test lower-bound]\label{prop:product-test-lower}
For $n = 2$ and $\omega \geq \tfrac{1}{2}$, consider the state $\ket{\psi} = \sqrt{\omega}\ket{11} + \sqrt{1-\omega}\ket{22}$.
Then $\operatorname{Overlap}_1(\ket{\psi}) = \omega$ and
\begin{equation*}
\operatorname{PT}_2(\ket{\psi}) = \omega^2 - \omega + 1.
\end{equation*}
In addition, for $n > 2$, consider $\ket{\psi} \otimes \ket{\phi}$,
where $\ket{\phi}$ is any product state in $\bbC^{d_3} \ot \cdots \otimes \bbC^{d_n}$.
Then this has the same overlap and probability of success as $\ket{\psi}$.
\end{prop}

The proof of \propref{product-test-lower} is standard and we include it in~\secref{prod-tight}.
Combining \thmref{product-test-main} and \propref{product-test-lower}
allows us to exactly compute $\operatorname{PT}(\omega)$ for $\omega \geq \tfrac{1}{2}$.

\begin{cor}[Product test, tight bound]
For all $n \geq 2$ and $\omega \geq \frac{1}{2}$, $\operatorname{PT}_n(\omega) = \omega^2 - \omega + 1$.
\end{cor}

This settles the performance of the product test when $\omega \geq \tfrac{1}{2}$.
The regime of $\omega < \tfrac{1}{2}$ remains open, however.
As~\cite{Harrow2013_product_testing} points out,
this regime ``is generally somewhat mysterious'',
and getting a better understanding of this case is part of open problem no.\ 2 in their work.
One possible starting point is to understand the behavior of $\operatorname{PT}_n(\omega)$ as $\omega \rightarrow 0$.
For example, as~\cite{Harrow2013_product_testing} show on page~$32$,
the $d$-dimensional maximally entangled state
$\ket{\psi} = \tfrac{1}{\sqrt{d}}\sum_{i=1}^d \ket{ii}$
has
\begin{equation*}
\omega = 1/d \quad \text{and} \quad \operatorname{PT}_2(\ket{\psi}) = \tfrac{1}{2}(1 + \tfrac{1}{d}).
\end{equation*}
This suggests the following question: does $\operatorname{PT}_n(\omega)\rightarrow \tfrac{1}{2}$ as $\omega \rightarrow 0$?

\begin{figure}[h]
\centering
\definecolor{lightblue}{RGB}{203,192,255}
\begin{tikzpicture}[
declare function={
    func(\x)=
    (\x<=1/2) * (1/3*\x*\x +2/3)   +
     (\x>1/2) * (\x*\x -  \x + 1);
  gunc(\x)=
  (\x >= 1-0.37013) * (1 - (1-\x) + (1-\x)^2 + (1-\x)^(3/2)) +
  (\x < 1-0.37013) * (1 - 11/512 * (1-\x));
  }]
\begin{axis}[
    	axis line style = thick,
	axis lines = left,
	axis line style={-},
	width=15cm,height=8cm,
	ymax = 1,
	xtick={0,0.125,0.25,0.5,1},
	xticklabels={$0$,$1/8$,$1/4$,$1/2$,1},
	xlabel={value of $\omega$},
	ytick={0.66667,1},
	yticklabels={$2/3$,$1$},
	ylabel={bounds on $\pt(\omega)$},
	every tick/.style=thick,
	every axis plot/.style={thick}
]
\addplot [
	domain=0:1,
	samples=100,
	color=lightblue,
	line width=4pt
	]
{gunc(x)};
\addplot [
	domain=0:1,
	samples=100,
	color=cyan
	]
{1-11/512*(1-x)};
\addplot [
	domain=0:1,
	samples=100,
	color=blue
	]
{1-(1-x)+(1-x)^2+(1-x)^(3/2)};
\addplot [
	domain=0:1,
	samples=100,
	color=pink,
	line width=4pt
	]
{func(x)};
\addplot [
	domain=0:1,
	samples=100,
	color=magenta
	]
{x^2-x+1};
\addplot [
	domain=0:1,
	samples=100,
	color=red
	]
{1/3*x^2+2/3};
\end{axis}
\end{tikzpicture}
  \captionsetup{singlelinecheck=off}
\caption[hello]{Upper bounds on $\pt(\omega)$
		as a function of $\omega = 1-\eps$.
	\begin{itemize}
	\item[$\circ$] The \textcolor{red}{red} line is the function $\tfrac{1}{3} \omega^2 + \tfrac{2}{3}$
			and the \textcolor{magenta}{magenta} line is the function $\omega^2 - \omega+1$.
			The thick \textcolor{pink}{pink} line is the minimum of the two.
			This is the upper bound we prove.
	\item[$\circ$] The \textcolor{blue}{blue} line is the function $1-\eps + \eps^2 + \eps^{3/2}$
			and the \textcolor{cyan}{cyan} line is the function $1 - \tfrac{11}{512} \eps$.
			The thick \textcolor{lightblue}{light blue} line is the minimum of the two.
			This is the upper bound of Harrow and Montanaro~\cite{Harrow2013_product_testing}.
	\end{itemize}}
\label{fig:hm-vs-us}
\end{figure}
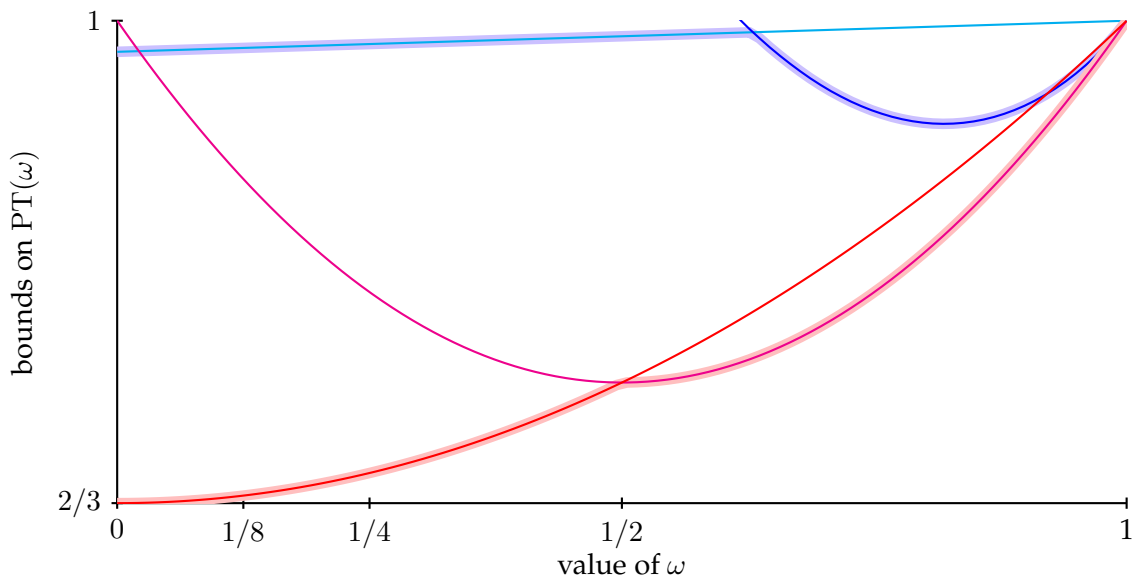
 
Our proof of \thmref{product-test-main} is a simple inductive argument.
Decomposing the product test measurement as 
$\piprod = (\iden \otimes \piswap^{\otimes n-1}) \cdot (\piswap \otimes \iden)$,
we can view it as first performing the SWAP test on the first qudit register of $\ket{\psi}$ and then, if it succeeds,
performing the $(n-1)$-qudit product test on the remaining qudit registers.
Supposing that $\ket{\psi}$ is far from being a product state, 
either the first qudit of $\ket{\psi}$ is highly entangled with the remaining qudits,
or the other qudits are far from being a product state (even conditioned on the first SWAP test succeeding).
In the first case, the SWAP test rejects with good probability,
and in the second case, the $(n-1)$-qubit products test rejects with good probability, by induction.
Balancing between these two cases gives our bound.

The proof of the bound 
$\pt_{n}(\omega)\leq \tfrac{1}{3} \omega^2 + \tfrac{2}{3}$
is especially simple and fits in a page.
Though weaker than our general bound when $\omega \geq \tfrac{1}{2}$,
this bound is still sufficient to recover all the applications
of the product test in~\cite{Harrow2013_product_testing},
including the proof that $\QMA(2) = \QMA(k)$ for $k \geq 2$.
We include it as a separate argument in \secref{prod-simple}.
The proof of the general bound from \thmref{product-test-main} is contained in \secref{prod-tight}.

So far we have considered the case of product testing where the number of copies~$m$ is exactly two,
but the property testing model requires us to take~$m$ sufficiently large to detect non-product states with constant probability.
For even $m$, a simple strategy is to run $m/2$ parallel copies of the product test and reject if any of them rejects.
If $\overlap_1(\ket{\psi}) = 1-\eps$, then this will accept with probability at most
$
(1-\tfrac{2}{3}\eps + \tfrac{1}{3} \eps^2)^{m/2}
$.
Making this probability smaller than $\tfrac{1}{3}$ as required by \defref{tester} entails setting $m = O(1/\eps)$.
Using the distance $\delta = \sqrt{\eps}$,
this can be stated as follows.

\begin{prop}[Copy complexity of testing product states]
Following the language of \defref{tester}, testing whether a state $\ket{\psi} \in \bbC^{d_1} \otimes \cdots \otimes \bbC^{d_n}$ is a product state can be done using $m = O(1/\delta^2)$ copies and with prefect completeness.

\end{prop}
\noindent
This is optimal, as $\Omega(1/\delta^2)$ copies are always required to distinguish between two states which are $\delta$-far from each other in trace distance.
We note that the same copy complexity follows from \thmref{hm}, the bound given by Harrow and Montanaro~\cite{Harrow2013_product_testing}.

\subsection{Testing matrix product states}

Having already considered the case of MPS testing with bond dimension $r = 1$,
we now consider the case of bond dimension $r > 1$.
To our knowledge, there is no prior work on this problem.

One idea for testing MPS is to use the general ``test-by-learning'' framework from property testing.
In our case, given a state $\ket{\psi}$,
this entails performing MPS tomography on $\ket{\psi}$ to learn an $\mps(r)$ approximation $\ket{\phi}$
and then applying the SWAP test on $\ket{\psi}$ and $\ket{\phi}$.
If $\ket{\psi}$ is in $\mps(r)$, then $\ket{\phi}$ will be a good approximation, and so the SWAP test will usually succeed,
but if $\ket{\psi}$ is far from $\mps(r)$, then $\ket{\phi}$ will be a bad approximation, and so the SWAP test will usually fail.
Various algorithms for MPS tomography have been proposed in the literature,
for example those in the works~\cite{cramer2010efficientTomographyMPS,lanyon2017efficientTomographyMPS}.
One would expect that since states in $\mps(r)$ can be described using $n d r^2$ parameters,
where $d$ is the largest subsystem dimension,
the optimal algorithm for $\mps(r)$ tomography should use $O(n d r^2/\delta^2)$ copies, 
though this precise bound is not yet known to our knowledge.
We propose and analyze a more direct MPS testing algorithm that improves on this ``test-by-learning'' method by a factor of $O(d)$.  

 We begin by designing an algorithm for this problem which we call the \emph{MPS tester}.
The MPS tester is motivated by the fact that $\ket{\psi_{1, \ldots, n}}$ is in $\mps(r)$
if and only if $\psi_{1, \ldots, k}$ has rank~$r$ for each $1 \leq k \leq n$.
This relates the problem of MPS testing
to the problem of \emph{rank testing},
i.e.\ of testing whether a mixed state~$\rho$ has rank~$r$,
which was previously considered in the work of O'Donnell and Wright~\cite{Wright2015_spectrum_testing}.
They designed an algorithm called the \emph{rank tester}
which can test whether~$\rho$ is rank~$r$ using $m = \Theta(r^2/\delta)$ copies of~$\rho$. 
When the $r = 1$ rank tester is run with $m = 2$ copies of~$\rho$,
it is equivalent to the SWAP test,
and for larger values of~$r$ and $m$ it uses a generalization of the SWAP measurement
known as \emph{weak Schur sampling}.
It has perfect completeness,
meaning that it always accepts states of rank~$r$,
and in fact it is the optimal test for states of rank~$r$ with perfect completeness,
as shown in~\cite[Proposition 6.1]{Wright2015_spectrum_testing}.

With the rank tester in hand,
we define the MPS tester to be the algorithm which
simultaneously performs a separate instance of the rank tester on $\psi_{1, \ldots, k}$ for each $1 \leq k \leq n$
and accepts if each instance of the rank tester accepts.
We include an illustration of the MPS tester in~\fig{mps-tester}.
We show that this test has perfect completeness,
meaning that it accepts every state in $\mps(r)$ with probability~$1$,
although we are not sure if it is the optimal algorithm with perfect completeness;
we view this as an interesting open direction.
We show the following bound on its copy complexity.

\begin{thm}[Copy complexity of the MPS tester]\label{thm:mps-tester}
Given $m = O(n r^2/\delta^2)$ copies of a state $\ket{\psi} \in \bbC^{d_1} \otimes \cdots \otimes \bbC^{d_n}$,
the MPS tester tests whether $\ket{\psi}$ is in $\mps(r)$ with perfect completeness.
\end{thm}

To prove this result, we first show that if $\ket{\psi}$ is $\delta$-far from the set $\mps(r)$,
then there exists $1 \leq k \leq n$ such that $\psi_{1, \ldots, k}$ is $\delta' = (\delta^2/2n)$-far from being rank-$r$.
Then the probability that the MPS tester accepts $\ket{\psi}$
is at most the probability that the rank tester accepts $\psi_{1, \ldots, k}$,
and this is at most $1/3$ given that we are using $O(r^2/\delta') = O(n r^2/\delta^2)$ copies of~$\ket{\psi}$.
One minor technicality that arises is checking that the MPS tester does indeed perform a valid measurement,
which entails showing that the rank testers for each $\psi_{1, \ldots, k}$ can all be simultaneously measured.

\begin{rem}[Time complexity of the MPS tester]
The $m$-copy rank tester of~\cite{Wright2015_spectrum_testing} can be performed efficiently with a quantum circuit of size $\poly\left(m,\log(d)\right)$ 
using the algorithm of~\cite{Krovi2019efficienthigh} or \cite[Page 160]{harrow2005thesis} that implements weak Schur sampling.
Here $d$ is the dimension of the state $\rho$  whose $m$ copies $\rho^{\ot m}$ are input to the rank tester.
The MPS tester performs the rank tester on $m=nr^2/\d^2$ copies of the reduced states $\psi_{1,\dots,k}$ for $1\leq k\leq n$.  
The maximum dimension of these reduced states is less than $\log(d_1\dots d_n)\leq n\log(d)$ where $d=\max_{i\in[n]}d_i$. 
Hence, the MPS tester can be implemented with a quantum circuit of size~$\poly\left(n,r,1/\d,\log(d)\right)$.
\end{rem}

We believe that the bound in \thmref{mps-tester} is not tight,
and that an inductive argument similar to our analysis of the product tester
should be able to improve it.
As an example, consider the ``bunny state''
\begin{equation*}
\ket{b_n} = \tfrac{1}{\sqrt{n - 1}}(\ket{110 \cdots 0} + \ket{0 11 \cdots 0} + \cdots + \ket{0 \cdots 0 11}).
\end{equation*}
We can show that this state, which is in $\mps(3)$, has $\overlap_2(\ket{b_n}) \leq \tfrac{2}{3}$. But does the $r = 2$ MPS tester detect this?
The above analysis suggests we should find the reduced density matrix $(b_n)_{1, \ldots, i}$
which is farthest from being rank~$2$.
It can be checked that $(b_n)_1$ and $(b_n)_{1, \ldots, n-1}$ are both rank~$2$.
Otherwise, for $i \in \{2, \ldots, n-2\}$,
the Schmidt decomposition of $\ket{b_n}$
into subsystems $\{1, \ldots, i\}$ and $\{i+1, \ldots, n\}$ is
\begin{equation*}
\ket{b_n} = \sqrt{\tfrac{i-1}{n-1}} \ket{b_i} \otimes \ket{0 \cdots 0}
	+ \tfrac{1}{\sqrt{n-1}} \ket{0 \cdots 0 1} \otimes \ket{1 0 \cdots 0}
	+ \sqrt{\tfrac{n-i-1}{n-1}} \ket{b_{n-i}} \otimes \ket{0 \cdots 0}
\end{equation*}
Hence, $(b_n)_{1, \ldots, i}$ has eigenvalues $\tfrac{i-1}{n-1}$, $\tfrac{1}{n-1}$, $\tfrac{n-i-1}{n-1}$,
and so it is distance $\tfrac{1}{n-1}$ from rank-$2$.
As a result, the rank tester needs $O(n)$ copies of $(b_n)_{1, \ldots, i}$ to detect this,
and therefore the MPS tester needs $O(n)$ copies of $\ket{b_n}$ if we use our above analysis.
However, we have done a more careful analysis of the bunny state
in line with the inductive argument for the product test,
and we can show that the MPS tester only needs 3 copies of $\ket{b_n}$ to detect that it is not in $\mps(2)$.
In particular, the MPS tester rejects $\ket{b_n}^{\otimes 3}$ with probability at least~$\tfrac{1}{6}$.
This means that the above analysis is too pessimistic, at least for the bunny state.

Unfortunately, we were unable to carry this proof strategy out in general.
One difficulty is that we are not even sure what upper bound to conjecture for this problem.
Originally, we had guessed that the MPS tester only needed $m = O(r^2/\delta^2)$ copies, or perhaps some other copy complexity which is independent of~$n$,
but we now know this is false, due to the following lower bound.

\begin{thm}[MPS testing lower bound]\label{thm:mps-lower}
For $r \geq 2$ and $\delta\leq 1/\sqrt{2}$,
testing whether a state $\ket{\psi} \in \bbC^{d_1} \otimes \cdots \otimes \bbC^{d_n}$, is in $\mps(r)$
requires $\Omega(n^{1/2}/\delta^2)$ copies of $\ket{\psi}$.
\end{thm}

\thmref{mps-lower} shows that a polynomial dependence on~$n$,
as in \thmref{mps-tester},
\emph{is} required,
even for the case of bond dimension $r = 2$.
This in sharp contrast to the $r = 1$ case of product testing,
in which a constant number of copies suffice, independent of~$n$.
This leaves open the following question: what is the optimal copy complexity for $\mps(r)$ testing, for $r \geq 2$?

The proof of the lower bound consists of two parts. We consider a quantum state~$\ket{\Phi_n}=\ket{\varphi}^{\ot \frac{n}{2}}$ where $\ket{\varphi}\in \bbC^d \ot \bbC^d$ for some $d\geq 2r-1$ and  $\distance_r\left(\ket{\varphi}\right)=\Omega(\d/\sqrt{n})$. In the first step, we use an inductive argument to prove that $\distance_r\left(\ket{\Phi_n}\right)\geq \d$.  In the second step, we consider the density matrix corresponding to the ensemble of states obtained by applying random local unitaries to the subsystems of $\ket{\Phi_n}$. Since all the states in this ensemble are $\d$-far from $\mps(r)$, a tester should reject this density matrix with probability at least $2/3$. We show that without sufficiently large number of copies, no $\mps(r)$ tester that accepts states in $\mps(r)$ with probability $\geq 2/3$ can also reject this density matrix with probability $\geq 2/3$.  

We prove our MPS tester upper bound (\thmref{mps-tester}) in \secref{testing algorithm}
and our MPS testing lower bound (\thmref{mps-lower}) in \secref{lower}.

\section{The product test}\label{sec:two-copy product test}

\subsection{A simple analysis of the product test}\label{sec:prod-simple}

We begin with a simple analysis of the product test
which shows that it rejects with constant probability if $\ket{\psi}$ is a constant distance from the set of product states.
This is sufficient to show $\QMA(2) = \QMA(k)$ via the proof of~\cite{Harrow2013_product_testing}.

\begin{thm}[Product test, simple bound]\label{thm:product test}
For all $n\geq 1$, $\operatorname{PT}_{n}(\omega)\leq \frac{1}{3} \omega^2 +\frac{2}{3}$.
\end{thm}

\begin{proof}
By induction, the $n = 1$ case being trivial. For the inductive step, let us assume \thmref{product test}
holds for $(n-1)$-partite states. Let $\ket{\psi} \in \bbC^{d_1} \ot \cdots \otimes \bbC^{d_n}$ be a state with $\operatorname{Overlap}(\ket{\psi}) = \omega$. For shorthand we write $d := d_1$. Note that the product test measurement can be written as
$\piswap^{\otimes n}=\left(\iden \otimes \piswap^{\otimes n-1}\right) \cdot\left(\piswap \otimes \iden\right)$. We can therefore view the test as first applying $\piswap$ to the first subsystem, and, if it succeeds, then applying $\piswap^{\otimes n-1}$ (i.e.~the product test) to the resulting
reduced state on the last $n-1$ subsystems. The probability this succeeds we bound by induction.

We begin by taking the Schmidt decomposition of $\ket{\psi}$ into subsystems $\{1\}$ and $\{2,\dots,n\}$:
\begin{align*}
|\psi\rangle=\sqrt{\lambda_{1}}\left|a_{1}\right\rangle\left|b_{1}\right\rangle+\cdots+\sqrt{\lambda_{d}}\left|a_{d}\right\rangle\left|b_{d}\right\rangle,
\end{align*}
where $\lambda_{1} \geq \cdots \geq \lambda_{d}$, $\left|a_{i}\right\rangle \in \mathbb{C}^{d},$ and $\left|b_{i}\right\rangle \in \mathbb{C}^{d_{2}} \otimes \cdots \otimes \mathbb{C}^{d_{n}}$. As a result,
\begin{align*}
|\psi\rangle^{\otimes 2}=\sum_{i \in[d]} \lambda_{i}\left|a_{i}\right\rangle^{\otimes 2}\left|b_{i}\right\rangle^{\otimes 2}+\sum_{i<j} \sqrt{\lambda_{i} \lambda_{j}}\left(\left|a_{i} a_{j}\right\rangle\left|b_{i} b_{j}\right\rangle+\left|a_{j} a_{i}\right\rangle\left|b_{j} b_{i}\right\rangle\right).
\end{align*}
The result of applying the first projector to $\ket{\psi}^{\ot 2}$ is therefore
\begin{align}\label{eq:p1}
\piswap \otimes I \cdot|\psi\rangle^{\otimes 2}=\sum_{i \in[d]} \lambda_{i}\left|a_{i}\right\rangle^{\otimes 2}\left|b_{i}\right\rangle^{\otimes 2}+\sum_{i<j} \sqrt{\lambda_{i} \lambda_{j}}\left(\frac{\left|a_{i} a_{j}\right\rangle+\left|a_{j} a_{i}\right\rangle}{\sqrt{2}}\right)\left(\frac{\left|b_{i} b_{j}\right\rangle+\left|b_{j} b_{i}\right\rangle}{\sqrt{2}}\right).
\end{align}
We note that this vector's two-norm, and hence the probability the test passes in the first step, is  $\mu:=\sum_{i} \lambda_{i}^{2}+\sum_{i<j} \lambda_{i} \lambda_{j}$. Conditioned on this, the mixed state of subsystems $2, \ldots, n$ is
\begin{align*}
\left|b_{i}\right\rangle^{\otimes 2} \text { with prob. } \frac{\lambda_{i}^{2}}{\mu}, \quad \frac{1}{\sqrt{2}}\left(\left|b_{i} b_{j}\right\rangle+\left|b_{j} b_{i}\right\rangle\right) \text { with prob. } \frac{\lambda_{i} \lambda_{j}}{\mu}.
\end{align*}
This is by \eqref{eq:p1} and the fact that the $\ket{a_i}^{\otimes 2}$'s and the $\left(\left|a_{i} a_{j}\right\rangle+\left|a_{j} a_{i}\right\rangle\right)$'s are orthogonal. We must now bound the probability that the product test on $n-1$ subsystems succeeds in each of these cases. In the first case this is $\mathrm{PT}_{n-1}\left(\left|b_{1}\right\rangle\right)$, and in the rest of the cases we will charitably bound the probability by $1$. This gives us:
\begin{align}\label{eq:to-be-used-later}
\mathrm{PT}_{n}(|\psi\rangle) \leq \mu \cdot\left(\frac{\lambda_{1}^{2}}{\mu} \cdot \mathrm{PT}_{n-1}\left(\left|b_{1}\right\rangle\right)+\sum_{i>1} \frac{\lambda_{i}^{2}}{\mu}+\sum_{i<j} \frac{\lambda_{i} \lambda_{j}}{\mu}\right)=\lambda_{1}^{2} \cdot \mathrm{PT}_{n-1}\left(\left|b_{1}\right\rangle\right)+\sum_{i>1} \lambda_{i}^{2}+\sum_{i<j} \lambda_{i} \lambda_{j}.
\end{align}
 Writing $\phi=\operatorname{Overlap}\left(\left|b_{1}\right\rangle\right)$, the inductive hypothesis gives us
\begin{align}
\mathrm{PT}_{n}(|\psi\rangle) &\leq \lambda_{1}^{2} \cdot\left(\frac{1}{3} \phi^{2}+\frac{2}{3}\right)+\sum_{i>1} \lambda_{i}^{2}+\sum_{i<j} \lambda_{i} \lambda_{j}\nonumber\\
&=\left(\frac{1}{3}\left(\lambda_{1} \phi\right)^{2}+\frac{2}{3}\right)-\frac{1}{3} \sum_{i>1}\left(\lambda_{1}-\lambda_{i}\right) \lambda_{i}-\frac{1}{3} \sum_{1<i<j} \lambda_{i} \lambda_{j} \leq \frac{1}{3}\left(\lambda_{1} \phi\right)^{2}+\frac{2}{3},\label{eq:p2}
\end{align}
where the last inequality follows because $\lambda_{1} \geq \cdots \geq \lambda_{d}$. Now, by definition of $\phi$, there exists a product state $|v\rangle \in \mathbb{C}^{d_2} \otimes \cdots \otimes \mathbb{C}^{d_{n}}$ such that $\phi=\left|\left\langle b_{1} | v\right\rangle\right|^{2}$. But then $\left|a_{1}\right\rangle \otimes|v\rangle$, also a product state, has squared-inner-product $\lambda_{1} \phi$ with $|\psi\rangle$, meaning $\lambda_{1} \phi \leq$ $\operatorname{Overlap}(\psi)=\omega$, and so by \eqref{eq:p2} the test succeeds with probability at most $\frac{1}{3} \omega^{2}+\frac{2}{3}$.\qedhere
 \end{proof}

\subsection{A tight analysis of the product test for $\omega \geq \tfrac{1}{2}$}\label{sec:prod-tight}

Next, we sharpen our upper-bound from \thmref{product test} in the $\omega \geq \tfrac{1}{2}$ case.

\begin{thm}[Product test, sharpened bound; \thmref{product-test-main} restated]\label{thm:product-test-tight}
For all $n\geq 1$,
\begin{equation*}
\operatorname{PT}_{n}(\omega)\leq
\left\{\begin{array}{cl}
\omega^2 - \omega + 1 & \text{if $\omega \geq \tfrac{1}{2}$},\\
\tfrac{1}{3} \omega^2 + \tfrac{2}{3} & \text{otherwise.}
\end{array}
\right.
\end{equation*}
\end{thm}

We begin by showing that \thmref{product-test-tight} is tight for $\omega \geq \tfrac{1}{2}$ using a simple example from \cite[Page 31]{Harrow2013_product_testing}).

\begin{prop}[Product test lower-bound; \propref{product-test-lower} restated]\label{prop:product-test-lower-restated}
For $n = 2$ and $\omega \geq \tfrac{1}{2}$, consider the state $\ket{\psi} = \sqrt{\omega}\ket{11} + \sqrt{1-\omega}\ket{22}$.
Then $\operatorname{Overlap}_1(\ket{\psi}) = \omega$ and
\begin{equation*}
\operatorname{PT}_2(\ket{\psi}) = \omega^2 - \omega + 1.
\end{equation*}
In addition, for $n > 2$, consider $\ket{\psi} \otimes \ket{\phi}$,
where $\ket{\phi}$ is any product state in $\bbC^{d_3} \ot \cdots \otimes \bbC^{d_n}$.
Then this has the same overlap and probability of success as $\ket{\psi}$.
\end{prop}
\begin{proof}
First, we show $\operatorname{Overlap}_1(\ket{\psi}) = \omega$.
This is because if $\ket{a} = \sum_{i=1}^{d_1} \alpha_i \ket{i}$
and $\ket{b} = \sum_{i=1}^{d_2} \beta_i \ket{i}$,
\begin{equation*}
|\braket{\psi}{a b}|^2
= |\sqrt{\omega}\cdot \alpha_1 \beta_1 + \sqrt{1-\omega}\cdot\alpha_2 \beta_2|^2
\leq \omega \cdot |\alpha_1 \beta_1|^2 + (1-\omega) \cdot|\alpha_2 \beta_2|^2.
\end{equation*}
This is maximized by taking $\alpha_1 = \beta_1 = 1$, in which case it equals~$\omega$.
Next, the probability of success is $\operatorname{PT}_2(\ket{\psi})=\Vert \piswap^{\otimes 2} \ket{\psi}^{\otimes 2}\Vert^2$,
 and so we first compute $\piswap^{\otimes 2} \ket{\psi}^{\otimes 2}$:
\begin{align*}
& \left(\frac{\iden + \mathrm{SWAP}}{2}\right)^{\otimes 2} \cdot (\sqrt{\omega}\ket{11} + \sqrt{1-\omega}\ket{22})^{\otimes 2}\\
={}& \left(\frac{\iden + \mathrm{SWAP}}{2}\right)^{\otimes 2}
	\cdot (\omega \ket{11} \ket{11} + \sqrt{\omega (1-\omega)} (\ket{11} \ket{22} +\ket{22} \ket{11}) + (1-\omega) \ket{22} \ket{22})\\
={}& \omega \ket{11} \ket{11}
	+ \sqrt{\omega(1-\omega)}\left(\frac{\ket{12}+\ket{21}}{\sqrt{2}}\right)\otimes \left(\frac{\ket{12}+\ket{21}}{\sqrt{2}}\right)
	+ (1-\omega) \ket{22} \ket{22}.
\end{align*}
The squared length of this is $\omega^2 + \omega(1-\omega) + (1-\omega)^2 = \omega^2-\omega+1$,
and so this equals $\operatorname{PT}_2(\ket{\psi})$.

The $n > 2$ case is an immediate consequence of the $n = 2$ case.
\end{proof}

Now we prove \thmref{product-test-tight}.

\begin{proof}[Proof of \thmref{product-test-tight}]
By induction, where again the $n = 1$ case is trivial. For the inductive step,
let $f_0(\omega) = \tfrac{1}{3} \omega^2 + \tfrac{2}{3}$ and $I_0 = [0, \tfrac{1}{2}]$. 
Let $f_1(\omega) = \omega^2-\omega + 1$ and $I_1 = [\tfrac{1}{2}, 1]$. 
The upper bound on $\operatorname{PT}(\omega)$ we are trying to show is
\begin{equation*}
\mathsf{UB}(\omega)=
\left\{\begin{array}{cl}
f_0(\omega) & \text{if $\omega \in I_0$},\\
f_1(\omega) & \text{if $\omega \in I_1$}.
\end{array}
\right.
\end{equation*}
Note that $\mathsf{UB}(\omega)$ is a non-decreasing function of~$\omega$,
and that $\mathsf{UB}(\omega) = \min\{f_0(\omega), f_1(\omega)\}$
because $f_0(\omega) \leq f_1(\omega)$ for $\omega \in I_0$
and $f_1(\omega) \leq f_0(\omega)$ for $\omega \in I_1$.
Recalling the proof of \thmref{product test}, we showed in \eqref{eq:to-be-used-later} that
\begin{align*}
\mathrm{PT}_{n}(|\psi\rangle)
\leq \lambda_{1}^{2} \cdot \mathrm{PT}_{n-1}\left(\left|b_{1}\right\rangle\right)+\sum_{i>1} \lambda_{i}^{2}+\sum_{i<j} \lambda_{i} \lambda_{j}.
\end{align*}
Write $\phi=\operatorname{Overlap}\left(\left|b_{1}\right\rangle\right)$,
and suppose that $\phi \in I_\alpha$, for $\alpha \in \{0, 1\}$.
Then the inductive hypothesis gives us
\begin{equation}\label{eq:slightly-more-general-inductive-hypothesis}
\mathrm{PT}_{n}(|\psi\rangle)
\leq \lambda_{1}^{2} \cdot f_\alpha(\phi) +\sum_{i>1} \lambda_{i}^{2}+\sum_{i<j} \lambda_{i} \lambda_{j}.
\end{equation}
Recall also that $\lambda_1 \phi \leq \operatorname{Overlap}\left(\ket{\psi}\right) = \omega$,
and suppose that $\lambda_1 \phi \in I_\beta$, for $\beta \in \{0, 1\}$.
Our goal will be to show the inequality $\eqref{eq:slightly-more-general-inductive-hypothesis} \leq f_\beta(\lambda_1 \phi)$.
Then because $\lambda_1\phi \in I_\beta$ we have $f_\beta(\lambda_1 \phi) = \mathsf{UB}(\lambda_1 \phi)$,
and because $\lambda_1 \phi \leq \omega$ and $\mathsf{UB}(\cdot)$ is a nondecreasing function,
we have $\mathsf{UB}(\lambda_1 \phi) \leq \mathsf{UB}(\omega)$, completing the inductive step.
Note that we only have to show $\eqref{eq:slightly-more-general-inductive-hypothesis} \leq f_\beta(\lambda_1 \phi)$
in the case that $\beta \leq \alpha$,
as $\lambda_1 \phi \leq \phi$, so we will never have $\beta > \alpha$.
In particular, we need not consider the case $\alpha = 0, \beta = 1$.
\newline

\noindent
\textit{Case 1: $\alpha = 1, \beta = 1$.} 
This case can be shown as follows.
\begin{align*}
&\lambda_{1}^{2} \cdot f_1(\phi)+\sum_{i>1} \lambda_{i}^{2}+\sum_{i<j} \lambda_{i} \lambda_{j}\\
={}& \lambda_1^2  \cdot (\phi^2 - \phi + 1) +\sum_{i>1} \lambda_{i}^{2}+\sum_{i<j} \lambda_{i} \lambda_{j}\\
={}& ((\lambda_1 \phi)^2 - \lambda_1 \phi + 1) -1 + \lambda_1 (1-\lambda_1) \phi +\sum_{i} \lambda_{i}^{2}+\sum_{i<j} \lambda_{i} \lambda_{j}\\
\leq{}& f_1(\lambda_1 \phi) -1 + \lambda_1 (1-\lambda_1)  +\sum_{i} \lambda_{i}^{2}+\sum_{i<j} \lambda_{i} \lambda_{j} \tag{because $\phi \leq 1$}\\
={}& f_1(\lambda_1 \phi) -1 + \sum_{1 < j} \lambda_1 \lambda_j  +\sum_{i} \lambda_{i}^{2}+\sum_{i<j} \lambda_{i} \lambda_{j} \\
\leq{}& f_1(\lambda_1 \phi) -1 + \sum_{i} \lambda_{i}^{2}+2 \sum_{i<j} \lambda_{i} \lambda_{j} \\
={}&f_1(\lambda_1 \phi) -1 + \Big(\sum_{i} \lambda_{i}\Big)^2 \\
={}& f_1(\lambda_1 \phi). \tag{because $\lambda_1 + \cdots + \lambda_d = 1$}
\end{align*}
This completes the proof.
\newline

\noindent
\textit{Case 2: $\alpha = 1, \beta = 0$.} 
This case can be shown as follows.
\begin{align*}
&\lambda_{1}^{2} \cdot f_1(\phi)+\sum_{i>1} \lambda_{i}^{2}+\sum_{i<j} \lambda_{i} \lambda_{j}\\
={}& \lambda_1^2  \cdot (\phi^2 - \phi + 1) +\sum_{i>1} \lambda_{i}^{2}+\sum_{i<j} \lambda_{i} \lambda_{j}\\
={}& (\tfrac{1}{3}(\lambda_1 \phi)^2 +\tfrac{2}{3}) -\tfrac{2}{3}
	+ \lambda_1^2 (\tfrac{2}{3} \phi^2 - \phi + \tfrac{1}{3}) +\tfrac{2}{3} \lambda_1^2+\sum_{i>1} \lambda_{i}^{2}+\sum_{i<j} \lambda_{i} \lambda_{j}\\
\leq{}& f_0(\lambda_1 \phi) -\tfrac{2}{3}+\tfrac{2}{3} \lambda_1^2
	+\sum_{i>1} \lambda_{i}^{2}+\sum_{i<j} \lambda_{i} \lambda_{j} \tag{because $\tfrac{1}{2} \leq \phi \leq 1$}\\
\leq{}& f_0(\lambda_1 \phi) -\tfrac{2}{3}+\tfrac{2}{3} \sum_i \lambda_i^2
	+\tfrac{1}{3} \sum_{1 < j} \lambda_{1} \lambda_j+\sum_{i<j} \lambda_{i} \lambda_{j} \tag{because $\lambda_j \leq \lambda_1$}\\
\leq{}& f_0(\lambda_1 \phi) -\tfrac{2}{3}+\tfrac{2}{3} \sum_i \lambda_i^2
	+\tfrac{4}{3} \sum_{i<j} \lambda_{i} \lambda_{j}\\
={}& f_0(\lambda_1 \phi) -\tfrac{2}{3}+\tfrac{2}{3} \Big(\sum_i \lambda_i\Big)^2\\
={}& f_0(\lambda_1 \phi). \tag{because $\lambda_1 + \cdots + \lambda_d = 1$}
\end{align*}
This completes the proof.
\newline

\noindent
\textit{Case 3: $\alpha = 0, \beta = 0$.} 
In this case, we aim to show that
\begin{equation*}
\lambda_{1}^{2} \cdot f_0(\phi)+\sum_{i>1} \lambda_{i}^{2}+\sum_{i<j} \lambda_{i} \lambda_{j}
\leq f_0(\lambda_1 \phi).
\end{equation*}
This was already shown in Equation~\eqref{eq:p2} in the proof of \thmref{product test} for all $\phi$ and $\lambda_1$.
This concludes case 3 and the proof of the theorem.
\end{proof}


\section{Preliminaries for MPS testing}\label{sec:Preliminaries}

\subsection{Low rank approximation to MPS}

\begin{lem}[Young-Eckart Theorem~\cite{eckart1936approximation}]\label{lem:Eckart}
Consider a bipartite state $\ket{\psi} \in \bbC^{d_1} \ot \bbC^{d_2}$ with $d_1 \geq d_2$,
and let
\begin{equation*}
\ket{\psi} =\sum_{i=1}^{d_2} \sqrt{\lambda_i} \ket{a_i}\ket{b_i}
\end{equation*}
be its Schmidt decomposition, where $\lambda_1 \geq \cdots \geq \lambda_{d_2}$.
Then the maximum overlap of $\ket{\psi}$  with a state in $\mps(r)$
is $\overlap_r\left(\ket{\psi}\right)=\sum_{i=1}^r \l_i$,
and it is achieved by the state
\begin{equation*}
\ket{\phi}=\frac{1}{\sqrt{\sum_{i=1}^r \l_i}}\sum_{i=1}^r\sqrt{\lambda_i} \ket{a_i}\ket{b_i}.
\end{equation*}
\end{lem}

While \lemref{Eckart} gives the closest $\mps(r)$  approximation to a bipartite state,
a general lower bound on the overlap between an $n$-partite state and $\mps(r)$ can be also derived. This is stated in the following lemma.
\begin{lem}[Low-rank Approximation,~Lemma 1 of~\cite{VerstraeteTruncationMPS}]\label{lem:low_rank_app}
Consider an $n$-partite state~$\ket{\psi} \in \bbC^{d_1} \otimes \cdots \otimes \bbC^{d_n}$.
For each $i \in \{1, \ldots, n-1\}$, write the  Schmidt decomposition of $\ket{\psi}$ across the subsystems $\{1, \ldots, i\}$ and $\{i+1, \ldots, n\}$ as
\begin{equation*}
\ket{\psi} = \sum_{j=1}^{D_i} \sqrt{\lambda^{(i)}_j} \ket{a^{(i)}_j} \ket{b^{(i)}_j},
\end{equation*}
where $D_i = \min\{d_1 \cdots d_i, d_{i+1} \cdots d_n\}$
and $\lambda_1^{(i)} \geq \cdots \geq \lambda_{D_i}^{(i)}$.
Then there exists a state $\ket{\phi} \in \mps(r)$ with (unsquared) overlap
\begin{equation*}
|\braket{\phi}{\psi}|\geq 1-\sum_{i=1}^{n-1} \sum_{j = r+1}^{D_i}\lambda^{(i)}_{j}.
\end{equation*}
\end{lem}

\subsection{Representation theory and weak Schur sampling}\label{sec:representation theory}
A common scenario involves having i.i.d.\ copies of a state~$\r$,
which in this work could be the state $\ket{\psi_{1, \ldots, n}}$ or one of its marginals.
These copies are invariant under permutation.
The spectrum of the state $\r$, including its rank, is also invariant under the action of any unitary operator $U$ that maps $\r$ to $U \r U^{\dagger}$.
In this section, we study these symmetries and discuss how we can exploit them in the analysis of our MPS tester.
The content of this section is already covered in detail in previous works (see \cite[Section 2.5]{Wright18}, \cite[Section 5.3]{harrow2005thesis}, or \cite[Section 2]{Childs_weak_schur_sampling} for example).
Here we briefly review these topics.

\begin{definition}[Partitions]\label{def:partitions}
A partition of $m$, denoted by $\m \vdash m$, is a list of nonnegative integers $\m=(\m_1,\m_2,\dots,\m_k)$ that satisfy
$\m_1 \geq \m_2 \geq \dots \geq \m_k$ and $\m_1+\m_2+\dots+\m_k=m$. We call the number of nonzero elements $\m_i$ in $\m$ the length of the partition and denote it by $\ell(\m)$.
\end{definition}
The group of all permutations of $\{1, \ldots, m\}$ is known as the \emph{symmetric group}, and we denote it by $\symm{m}$.
In addition, we denote the group of $d\times d$ unitary operators by $\GL{d}$.
Two natural representations of the groups $\symm{m}$ and $\GL{d}$ over the space $\cdn$ are given as follows.
\begin{align*}
\srep(\pi)\; |a_1\rangle \otimes |a_2\rangle \otimes \ldots \otimes |a_m \rangle
&= |a_{\pi^{-1}(1)}\rangle \otimes |a_{\pi^{-1}(2)}\rangle \otimes \ldots \otimes |a_{\pi^{-1}(m)} \rangle,\\
\GLrep(U)\; |a_1\rangle \otimes |a_2\rangle \otimes \ldots \otimes |a_m \rangle
&= \ \,\!(U |a_1\rangle) \;\!\otimes\, (U |a_2\rangle) \,\otimes\,\! \ldots \otimes (U |a_m \rangle),
\end{align*}
where $\{\bigotimes_{i=1}^m \ket{a_i}\}$ with $a_i\in [d]$ is a basis for $\cdn$, and $\pi \in \symm{m},\ U\in \GL{d}$.
The irreducible representations (\emph{irreps}) of the symmetric group, denoted $\sirrep_\m$, are indexed by partitions $\mu \vdash m$.
Similarly, the polynomial irreps of the unitary group, denoted $\GLirrep_\m^d$, are indexed by partitions $\mu$ with $\ell(\m) \leq d$.
The dimension of the symmetric group irrep $\sirrep_\m$ is denoted $\dim(\m)$,
and its corresponding character $\chi_\mu$ is given by $\chi_{\mu}(\pi) = \Tr[\srep(\pi)]$.

The representations $\srep(\pi)$ and $\GLrep(U)$ commute, meaning that $\srep(\pi)\GLrep(U) =\GLrep(U) \srep(\pi)$.
Hence, we can consider $\srep(\pi)\GLrep(U)$ as a representation of the direct product group $\symm{m}\times \GL{d}$.
\emph{Schur-Weyl duality}, stated as follows, establishes a strong connection between these representations. 

\begin{thm}[Schur--Weyl duality] \label{thm:Schur Weyl}
The space $(\bbC^d)^{\ot m}$ decomposes as
\begin{equation*}
\srep \GLrep \mathop{\cong}^{\symm{m} \times \GL{d}} \Op_{\substack{\m \vdash m\\ \ell(\m)\leq d}} \sirrep_\m \ot \GLirrep_\m^d.
\end{equation*}
In other words, there exist a unitary $\Usch \in \unitary{d^m}$ such that for all $\pi \in \symm{m}$ and $U \in \GL{d}$,
\ba
\schurbasis\srep(\pi)\GLrep(U)\schurbasis^\dagger
= \sum_{\substack{\m \vdash m\\ \ell(\m)\leq d}} \ketbra{\m}{\m}\ot\sirrep_\lambda(\pi)\ot\GLirrep_\lambda^d(U).\label{eq:schur-basis}
\ea
\end{thm}
The unitary operator $\Usch$ transforms the standard basis into a basis that called the \emph{Schur basis} and label by $\ket{\m}\ket{q_\mu}\ket{p_\mu}$.
In this basis,
\begin{align*}
\GLrep(U)\cdot \ket{\mu}\ket{q_\mu}\ket{p_\mu}&=\ket{\m} (\GLirrep_\mu^d(U)\ket{q_\mu})\ket{p_\mu},\\
\sirrep(\pi)\cdot \ket{\mu}\ket{q_\mu}\ket{p_\mu}&=\ket{\m} \ket{q_\mu}(\sirrep_\mu(\pi)\ket{p_\mu}).
\end{align*}
Because $\GLirrep_\mu^d$ is a polynomial irrep,
it is well-defined for any $d \times d$ matrix.
For example, when applied to invertible matrices it gives the $\mu$-irrep of the general linear group $\mathsf{GL}_d$.
We can also apply it to (possibly) non-invertible matrices, like the state~$\rho$.
In this case, if we set $\pi=e$ in \eqref{eq:schur-basis}, where $e$ is the identity permutation,
we see that the operator $\GLrep(\r)=\r^{\ot m}$ is block-diagonalized in the Schur basis.
\begin{cor}
Given a $d \times d$ density operator $\r$,
\ba
\Usch \r^{\ot m} \Usch=\sum_{\substack{\m \vdash m\\ \ell(\m)\leq d}} \ketbra{\m}{\m}\ot \iden_{\dim( \m)}\ot\GLirrep_\m^d(\r).\label{eq:p4}
\ea
\end{cor}
The equality \eqref{eq:p4} shows that there is a unitary $\Usch$ \emph{independent} of the state $\r$ which puts this state in the block-diagonal form.
We can therefore interpret the density matrix $\Usch \r^{\ot m} \Usch$
as corresponding to a mixed state with one element in the mixture for each block $\mu$.
In this case, measuring the block~$\mu$ can be done without loss of generality,
as it does not perturb the state.
This gives rise to the following measurement.
\begin{definition}[Weak Schur sampling]\label{def:WSS}
\emph{Weak Schur sampling (WSS)}
refers to the projective measurement $\{\Pi_\mu\}_{\mu \vdash m, \ell(\mu) \leq d}$
in which $\Pi_\mu$ projects onto the subspace specified by the partition $\mu$ in the Schur basis.
\end{definition}

The distribution of $\m$ measured by WSS only depends on the spectrum of the state $\r$. 
In fact, if one is only interested learning some property of $\r$'s spectrum,
it can be shown that WSS is the optimal measurement,
and that further measuring within the $\m$-irrep (e.g.\ measuring $\GLirrep_\m^d(\r)$)
yields no additional information about $\r$'s spectrum.

In the analysis of the MPS tester in \secref{testing algorithm}, we will use the following expression for the projector $\Pi_{\m}$ using the characters $\chi_{\m}$.
\begin{thm}[Weak Schur sampling projector, cf.\ {\cite[Equation 7]{Childs_weak_schur_sampling}}]\label{thm:representations of Pi_Lambda}
The weak Schur sampling projectors $\Pi_{\m}$ can be expressed as
\begin{align}
\Pi_{\m}=\dim(\m) \cdot \E_{\boldsymbol{\pi} \in \symm{m}}\left[ \chi_{\m}(\boldsymbol{\pi})\srep(\boldsymbol{\pi})\right]\label{eq:g11}.
\end{align}
\end{thm}

\section{An algorithm for testing matrix product states}\label{sec:testing algorithm}

In this section we introduce the MPS tester.
To begin, we introduce the rank tester of O'Donnell and Wright~\cite{Wright2015_spectrum_testing},
which is meant to test whether a mixed state~$\rho$ is rank~$r$.
This refers to the following problem.

\begin{definition}[Rank testing]
Given a mixed stated state $\rho \in \bbC^{d \times d}$,
let $\rho = \sum_{i=1}^d \alpha_i\cdot  \ket{u_i}\bra{u_i}$ be its eigendecomposition,
where $\alpha_1 \geq \cdots \geq \alpha_d$.
Then $\rho$ is \emph{$\d$-far from rank~$r$} if $\lambda_{r+1} + \cdots + \lambda_d \geq \d$.

An algorithm $\mathcal{A}$
is a \emph{property tester for rank $r$ matrices using $m = m(r, \d)$ copies}
if, given $\d > 0$ and $m$ copies of $\rho \in \bbC^{d \times d}$, it acts as follows.
If $\rho$ is rank-$r$, then it accepts with probability at least $\tfrac{2}{3}$.
(If instead it accepts with probability~$1$ in this case, we say that it has \emph{perfect completeness}.)
And if $\rho$ is $\d$-far from rank-$r$, then it accepts with probability at most $\tfrac{1}{3}$.
\end{definition}

The rank tester of~\cite{Wright2015_spectrum_testing} is motivated by the fact that if $\rho$ is indeed rank~$r$,
then weak Schur sampling (as in \defref{WSS}) always returns a Young diagram $\boldsymbol{\m}$ with $\ell(\boldsymbol{\m}) \leq r$.

\begin{definition}[The rank tester]
Let $r \geq 1$.
Given $\rho^{\otimes n}$, the \emph{rank tester}
performs weak Schur sampling and receives a random $\boldsymbol{\m}$.
The rank tester accepts if $\ell(\boldsymbol{\m}) \leq r$ and rejects otherwise.
Equivalently, it performs the two-outcome projective measurement $\{\Pi_{\leq r}, \iden-\Pi_{\leq r}\}$,
where $\Pi_{\leq r} = \sum_{\m: \ell(\m) \leq r} \Pi_\m$, and accepts if it observes the first outcome.
\end{definition}

The next theorem states the copy complexity of the rank tester.

\begin{thm}[{\cite[Lemma 6.2]{Wright2015_spectrum_testing}}]
The rank tester tests whether~$\rho$ has rank~$r$ with $O(r^2/\d)$ copies.
\end{thm}

O'Donnell and Wright also show that the rank tester requires $\Omega(r^2/\d)$ copies~\cite[Lemma 6.2]{Wright2015_spectrum_testing}.
The rank tester has perfect completeness,
and in fact it is the optimal algorithm for rank testing with perfect completeness~\cite[Proposition 6.1]{Wright2015_spectrum_testing}.
However, among algorithms with imperfect completeness,
the best known lower bound states that $\Omega(r/\d)$
are necessary~\cite[Theorem 1.11]{Wright2015_spectrum_testing}.
It remains an open question whether the rank tester is indeed the optimal algorithm for this task,
or whether it can be improved upon.

Now we state the MPS tester.
It is motivated by the fact that $\ket{\psi_{1, \ldots, n}}$ is in $\mps(r)$
if and only if $\psi_{1, \ldots, i}$ has rank~$r$ for each $i \in [n]$.

\begin{definition}[The MPS tester]\label{def:mps-tester}
Given $m$ copies of the state $\ket{\psi_{1, \ldots, n}} \in \bbC^{d_1} \ot \dots \ot \bbC^{d_n}$, the \emph{MPS tester} acts as follows.
For all $i \in [n]$, it runs the rank tester on $\psi_{1, \ldots, i}^{\otimes m}$. It accepts if each of them accepts, and rejects otherwise.

Equivalently, for each $i \in [n]$, let $\cH_{1, \ldots i}$ be the Hilbert space $\cH_{1, \ldots i} = \bbC^{d_1} \ot \cdots \ot \bbC^{d_i}$,
and define $\cH_{i+1, \ldots, n}$ analogously.
Let $\{\Pi_{\leq r, 1, \ldots, i}, \iden - \Pi_{\leq r, 1, \ldots, i}\}$ be the rank tester's measurement when performed on $\cH_{1, \ldots, i}^{\otimes m}$.
Then the MPS tester performs the two-outcome projective measurement $\{\Pi_{\mathrm{MPS}}, \iden - \Pi_{\mathrm{MPS}}\}$,
where
\begin{equation*}
\Pi_{\mathrm{MPS}} = \prod_{i=1}^n (\Pi_{\leq r, 1, \ldots, i} \otimes \iden_{\cH_{i+1, \ldots, n}^{\otimes m}}),
\end{equation*}
and accepts if it observes the first outcome.
\end{definition}

Before analyzing the copy complexity of the MPS tester,
we first show that it is well-defined.
In particular, we will show that the different rank tester measurements commute with each other,
which implies that they can be simultaneously measured
and that $\{\Pi_{\mathrm{MPS}}, \iden - \Pi_{\mathrm{MPS}}\}$ is indeed a two-outcome projective measurement,
as claimed in \defref{mps-tester}.
We first prove the following lemma, which shows that two overlapping weak Schur sampling measurements commute.

\begin{lem}[Overlapping weak Schur sampling commutes]\label{lem:commuting_WSS}
Consider a bipartite system with Hilbert space $\cH_L \ot \cH_R$, where $\cH_L = \bbC^{d_L},\cH_R = \bbC^{d_R}$.
Let $\{\Pi_{\l,L}\}$ and $\{\Pi_{\mu,LR}\}$ denote the weak Schur sampling measurements
when applied to $\cH_L^{\otimes m}$ and $\cH_{LR}^{\otimes m}$, respectively.
Then these two measurements commute,
meaning that for any two partitions $\lambda$ and~$\mu$,
\begin{equation*}
(\Pi_{\lambda, L} \otimes \iden_{\cH_R^{\otimes m}}) \cdot \Pi_{\mu, LR}
= \Pi_{\mu, LR} \cdot (\Pi_{\lambda, L} \otimes \iden_{\cH_R^{\otimes m}}).
\end{equation*}
 \end{lem}
\begin{proof}
Throughout this proof, we will omit the ``$ \iden_{\cH_R^{\otimes m}}$''
when writing $\Pi_{\lambda, L} \otimes \iden_{\cH_R^{\otimes m}}$
or 
$\srep_{L} \otimes \iden_{\cH_R^{\otimes m}}$,
for simplicity.

First, we note that $\srep_L(\pi) \cdot \srep_{LR}(\sigma) = \srep_{LR}(\sigma) \srep_L(\sigma^{-1} \pi \sigma)$.
To show this, let $\ket{\ell_1}, \ldots, \ket{\ell_m}$ be~$m$ standard basis vectors in $\cH_L$,
and let $\ket{r_1}, \ldots, \ket{r_m}$ be~$m$ standard basis vectors in $\cH_R$.
Then each $\ket{\ell_i}\otimes \ket{r_i}$ is a standard basis vector in $\cH_L \ot \cH_R$,
and so
\begin{align*}
&\srep_L(\pi) \cdot \srep_{LR}(\sigma) \cdot \ket{\ell_1 \cdots \ell_m} \ot \ket{r_1 \cdots r_m}\\
 ={}& \srep_L(\pi) \cdot \ket{\ell_{\sigma^{-1}(1)} \cdots \ell_{\sigma^{-1}(m)}} \ot \ket{r_{\sigma^{-1}(1)} \cdots r_{\sigma^{-1}(m)}}\\
 ={}& \ket{\ell_{\sigma^{-1}(\pi^{-1}(1))} \cdots \ell_{\sigma^{-1}(\pi^{-1}(m))}} \ot \ket{r_{\sigma^{-1}(1)} \cdots r_{\sigma^{-1}(m)}}\\
 ={}& \srep_{LR}(\sigma) \cdot \ket{\ell_{\sigma^{-1}(\pi^{-1}(\sigma(1)))} \cdots \ell_{\sigma^{-1}(\pi^{-1}(\sigma(m)))}} \ot \ket{r_1 \cdots r_m}\\
={}&\srep_{LR}(\sigma) \cdot \srep_L(\sigma^{-1} \pi \sigma) \cdot \ket{\ell_1 \cdots \ell_m} \ot \ket{r_1 \cdots r_m}.
\end{align*}
Extending this to all of $\cH_L \ot \cH_R$ via linearity proves the equality.
Next, we note that $\chi_{\lambda}(\sigma^{-1} \pi \sigma) = \chi_\lambda(\pi)$
because $\chi_\lambda(\cdot)$ is a class function.
Putting these together, we have
\begin{align*}
\Pi_{\l,L}\Pi_{\mu,LR}&=\dim(\l)\dim(\mu)\cdot\E_{\boldsymbol{\pi}, \boldsymbol{\sigma} \sim \symm{m}}\left[ \chi_{\l}(\boldsymbol{\pi})\chi_{\mu}(\boldsymbol{\sigma})\srep_L(\boldsymbol{\pi})\srep_{LR}(\boldsymbol{\sigma})\right]\nonumber\\
&=\dim(\l)\dim(\mu) \cdot\E_{\boldsymbol{\pi}, \boldsymbol{\sigma}\sim \symm{m}}\left[ \chi_{\l}(\boldsymbol{\sigma}^{-1} \boldsymbol{\pi} \boldsymbol{\sigma})\chi_{\mu}(\boldsymbol{\sigma})\srep_{LR}(\boldsymbol{\sigma})\srep_{L}(\boldsymbol{\sigma}^{-1} \boldsymbol{\pi} \boldsymbol{\sigma}) \right]\nonumber\\
&=\dim(\l)\dim(\mu)\cdot \E_{\boldsymbol{\pi}, \boldsymbol{\sigma}\sim \symm{m}}\left[ \chi_{\l}(\boldsymbol{\pi})\chi_{\mu}(\boldsymbol{\sigma})\srep_{LR}(\boldsymbol{\sigma})\srep_{L}(\boldsymbol{\pi}) \right]\nonumber\\
&=\Pi_{\mu,LR}\Pi_{\l,L},
\end{align*}
where the third line uses the fact that $\boldsymbol{\sigma}^{-1} \boldsymbol{\pi} \boldsymbol{\sigma}$ is distributed as a uniformly random element of $\symm{m}$,
even conditioned on the value of $\boldsymbol{\sigma}$.
This completes the proof.
\end{proof}

As an immediate corollary, we get that the MPS tester is well-defined.

\begin{prop}[The MPS tester is well-defined]
The matrices
\begin{equation*}\Pi_{\leq r, 1, \ldots, i} \otimes \iden_{\cH_{i+1, \ldots, n}^{\otimes m}}
\end{equation*}
commute for all $i \in [n]$.
As a result, the MPS tester measurement $\{\Pi_{\mathrm{MPS}}, \iden - \Pi_{\mathrm{MPS}}\}$ is a two-outcome projective measurement.
\end{prop}

Now we analyze the copy complexity of the MPS tester.
Because it runs a separate rank tester on each cut of $\ket{\psi_{1, \ldots, n}}$ simultaneously,
the outcome of one rank tester can affect the rank of the remaining cuts,
and therefore the outcomes of the remaining rank testers.
This complicates the analysis of this collective set of measurements.
Instead, we will do a pessimistic analysis and just show that the MPS tester does well on at least one cut.
This analysis uses the following proposition.

\begin{prop}[Far from MPS implies a cut is far from low-rank]\label{prop:far-from-mps}
Suppose $\ket{\psi_{1, \ldots, n}}$ is $\delta$-far from $\mps(r)$.
Then there exists an $i \in [n-1]$ such that $\psi_{1, \ldots, i}$ is $(\delta^2/2n)$-far from rank-$r$.
\end{prop}
\begin{proof}
$\distance_r(\ket{\psi_{1, \ldots, n}}) \geq \delta$ implies $\overlap_r(\ket{\psi_{1, \ldots, n}}) \leq 1 - \delta^2$.
By \lemref{low_rank_app}, there exists a state $\ket{\phi} \in \mps(r)$ such that
\begin{equation*}
|\braket{\phi}{\psi}|\geq 1-\sum_{i=1}^{n-1} \sum_{j = r+1}^{D_i}\lambda^{(i)}_{j}.
\end{equation*}
Then $|\braket{\phi}{\psi}| \leq \sqrt{1 - \delta^2} \leq 1 - \delta^2/2$. Rearranging, we have
\begin{equation*}
\delta^2/2
\leq 1 - \sqrt{1 - \delta^2}
\leq 1 - |\braket{\phi}{\psi}|
\leq \sum_{i=1}^{n-1} \sum_{j = r+1}^{D_i}\lambda^{(i)}_{j}
\leq n \cdot \max_{i \in [n-1]} \left\{\sum_{j = r+1}^{D_i}\lambda^{(i)}_{j}\right\}.
\end{equation*}
Letting $i$ be the maximizing coordinate,
this implies that $\psi_{1, \ldots, i}$ is $(\delta^2/2n)$-far from rank-$r$,
which completes the proof.
\end{proof}

There are two ways that \propref{far-from-mps} ``loses'' in going from $\ket{\psi_{1, \ldots, n}}$ being $\delta$-far to $\psi_{1, \ldots, i}$ being $(\delta^2/2n)$-far.
The first is the factor of $1/n$ which is unavoidable since we are ignoring all but one cut.
The second ``loss'' is the fact that $\delta$ is squared in the conclusion.
However, this turns out to just be a quirk in the different ways we measure distance to MPS and distance to rank-$r$.
For example, even for a bipartite state $\ket{\psi_{1, 2}}$, 
\lemref{Eckart} tells us that $\ket{\psi_{1, 2}}$ is $\delta$-far from $\mps(r)$
if and only if $\psi_1$ is $\delta^2$-far from rank-$r$.

Now we prove the main theorem of this section.

\begin{thm}\label{thm:upper_bound}
Given $m=O(nr^2/\delta^2)$ copies of a state $\ket{\psi_{1, \ldots, n}} \in \bbC^{d_1}\ot \dots \ot \bbC^{d_n}$,
the MPS tester tests whether $\ket{\psi}$ is in $\mps(r)$ with perfect completeness.
\end{thm}

\begin{proof}
If $\ket{\psi_{1, \ldots, n}}$ is in $\mps(r)$,
then $\psi_{1, \ldots, i}$ is rank-$r$ for each $i \in [n]$.
As a result, the rank tester applied to each cut always accepts because the rank tester has perfect completeness,
and so the MPS tester always accepts.
On the other hand, if $\rho$ is $\delta$-far from $\mps(r)$,
then \propref{far-from-mps} implies there exists an $i \in [n-1]$
such that $\psi_{1, \ldots, i}$ is $\delta' = (\delta^2/2n)$-far from rank $r$.
The probability the MPS tester accepts $\ket{\psi_{1, \ldots, n}}$
is at most the probability the rank tester accepts $\psi_{1, \ldots, i}$,
and since we are using $O(nr^2/\delta^2) = O(r^2/\delta')$ copies,
the rank tester will accept with probability at most  $\tfrac{1}{3}$.
Thus, the MPS tester tests whether $\ket{\psi_{1, \ldots, n}}$ is in $\mps(r)$, and this completes the proof.
\end{proof}

\section{A lower bound for testing matrix product states}\label{sec:lower}
We now derive a lower bound on the sample complexity of testing whether a state is in $\mps(r)$, for $r \geq 2$.
Let $d\geq 0$ satisfy $d-1 \geq 2\cdot (r-1)$, and consider the bipartite state $\ket{\varphi}\in \bbC^d \ot \bbC^d$ defined as
\begin{equation*}
\ket{\varphi}=\sqrt{1 - \theta} \cdot \ket{1}\ket{1}+\sum_{i=2}^{d} \sqrt{\frac{\t}{d-1}}\cdot \ket{i}\ket{i},
\end{equation*}
where $0 \leq \theta\leq 1$ is a parameter that we set later.
By the Young-Eckart Theorem (\lemref{Eckart}),
\begin{equation*}
\overlap_r(\ket{\varphi}) = (1-\theta) + (r-1) \cdot \frac{\theta}{d-1}.
\end{equation*}
Let $n$ be an even integer, and define
\begin{equation*}
\ket{\Phi_n}=\ket{\varphi}^{\otimes \frac{n}{2}}.
\end{equation*}
To compute the overlap of $\ket{\Phi_n}$ with $\mps(r)$, we use the following proposition.

\begin{prop}[Overlap of tensor products]\label{prop:tensor-overlap}
Let $\ket{\varphi} \in \bbC^{d_1} \ot \cdots \ot \bbC^{d_k}$ be a $k$-partite state
with $\overlap_r(\ket{\varphi}) = \omega$.
Then for each $\ell \geq 1$, $\overlap_{r}(\ket{\varphi}^{\otimes \ell}) = \omega^{\ell}$.
\end{prop}
\begin{proof}
Let $\ket{\psi}$ be the state in $\mps_k(r)$ with $|\braket{\psi}{\varphi}|^2 = \omega$
guaranteed by the assumption. Then $\ket{\psi}^{\otimes \ell}$ is in $\mps_{\ell k}(r)$
and has $|\bra{\psi}^{\otimes \ell} \cdot \ket{\varphi}^{\otimes \ell}|^2 = |\braket{\psi}{\varphi}|^{2\ell} = \omega^{\ell}$.
This proves the lower-bound $\overlap_{\ell k}(\ket{\psi}^{\otimes \ell}) \geq \omega^{\ell}$,

As for the upper-bound,
the proof is by induction on~$\ell$,
the base case being trivial.
For the inductive step,
write $\ket{\varphi^{\ell}}$ as shorthand for $\ket{\varphi}^{\otimes \ell}$,
and suppose that the inductive hypothesis holds for $\ket{\varphi^{\ell}}$,
i.e.\ that $\overlap_r(\ket{\varphi^{\ell}}) \leq \omega^\ell$.
Then we show that 
\ba
\forall\ket{\b}\in \mps_{(\ell + 1) k}(r),\quad |\braket{\varphi^{\ell+1}}{\b}|^2\leq \omega^{\ell+1}.\nn
\ea
Since $\ket{\varphi^{\ell+1}}=\ket{\varphi^{\ell}}\ot \ket{\varphi}$, this is equivalent to proving that for all $\ket{\b}\in \mps_{(\ell+1)k}(r)$,
\ba
\braket{\varphi^{\ell+1}}{\b}\braket{\b}{\varphi^{\ell+1}}&=\bra{\varphi}\cdot(\bra{\varphi^{\ell}} \ot \iden) \cdot \ketbra{\b}{\b} \cdot (\ket{\varphi^{\ell}} \ot \iden) \cdot \ket{\varphi} \nn\\
&=\norm{\ket{\G}}^2 \cdot |\braket{\varphi}{\widetilde{\G}}|^2\leq\omega^{\ell+1}, \label{eq:g34}
\ea
where we define
\ba
\ket{\G}=(\bra{\varphi^{\ell}} \ot \iden) \cdot \ket{\b},\quad \ket{\widetilde{\G}}=\frac{\ket{\G}}{\norm{\ket{\G}}}.\label{eq:gg7}
\ea
From here, our proof of \eqref{eq:g34} breaks into two steps: step 1, showing that $\norm{\ket{\G}}^2\leq \omega^{\ell}$,
and  step 2, showing that $|\braket{\varphi}{\widetilde{\G}}|^2\leq \omega$. We begin with the former.
\newline

\noindent
\textit{Step 1:  bounding $\norm{\ket{\G}}^2$.} Let the Schmidt decomposition of the state $\ket{\b}\in \mps_{(\ell+1) k}(r)$ across the subsystems $\{1,\dots, \ell k\}$ and $\{\ell k + 1,\cdots, (\ell+1) k\}$ be $\ket{\b}=\sum_{i=1}^r \sqrt{\mu_i} \ket{e_i}\ket{f_i}$. Then $\ket{\G}=\sum_{i=1}^r \sqrt{\mu_i} \braket{\varphi^{\ell}}{e_i}\ket{f_i}$, and
\begin{equation*}
\norm{\ket{\G}}^2=\sum_{i=1}^r \mu_i |\braket{\varphi^{\ell}}{e_i}|^2.
\end{equation*}
Suppose it holds that $\ket{e_i}$ is in $\mps_{\ell k}(r)$, for each $i$. Then the inductive hypothesis implies that
\ba
\norm{\ket{\G}}^2\leq \sum_{i=1}^r \mu_i \omega^{\ell}=\omega^{\ell}, \nn
\ea
which is the desired bound on $\norm{\ket{\G}}^2$. It remains to show that $\ket{e_i} \in\mps_{\ell k}(r)$.
Consider partitioning the $(\ell+1)k$ subsystems into $\{1,\dots,q\}$ and $\{q+1,\dots,(\ell+1)k\}$ for any integer $1\leq q \leq \ell k$. We have
\begin{align}
\Tr_{q+1,\dots,(\ell+1)k} \ketbra{\b}{\b}
&=\Tr_{q+1,\dots,(\ell+1)k} \left[\sum_{i,j=1}^r \sqrt{\m_i \m_j}\ketbra{e_i}{e_j}\ot \ketbra{f_i}{f_j}\right]\nn\\
&=\sum_{i,j=1}^r \sqrt{\m_i \m_j} \cdot\Tr_{q+1,\dots,(k+1)\ell} \left[\ketbra{e_i}{e_j}\ot \ketbra{f_i}{f_j}\right]\nn\\
&=\sum_{i=1}^r \m_i \cdot \Tr_{q+1,\dots,k\ell} \ketbra{e_i}{e_i}.\label{eq:gg9}
\end{align}
Equation \eqref{eq:gg9} is a sum of PSD operators and rank does not decrease by adding such operators.
Since $\ket{\b}\in \mps_{(\ell+1)k}(r)$, we have $\Tr_{q+1,\dots,(\ell+1)k}\ketbra{\b}{\b}\leq r$.
For this to happen, the rank of each of the terms in Equation \eqref{eq:gg9} must also be $\leq r$.
This implies $\ket{e_i}\in \mps_{\ell k} (r)$, as we claimed. 
\newline

\noindent
\textit{Step 2: bounding $|\braket{\varphi}{\widetilde{\G}}|^2$.}
By the base case of the induction, $\overlap_r(\ket{\varphi})\leq \omega$.
Thus, to complete step 2, it is sufficient to show that $\ket{\widetilde{\G}}$ is in $\mps_k(r)$.
Let $D = (d_1 \cdots d_k)^\ell$,
and let $\ket{A_1}, \ket{A_2},\dots, \ket{A_{D}}$ be an orthonormal basis for the first $\ell k$ qudits such that $\ket{A_1}=\ket{\varphi^{\ell}}$.
In addition, let $i \in \{1, \ldots, k\}$.
By tracing out the first $\ell k + i$ qudits in the state $\ket{\b}\in \mps_{(\ell + 1) k}(r)$ we get
\ba
\Tr_{1,\dots,\ell k + i} \ketbra{\b}{\b}
= \Tr_{\ell k+1, \dots, \ell k + i} \left[\sum_{i\in[D]} (\bra{A_i} \ot \iden) \cdot \ketbra{\b}{\b} \cdot (\ket{A_i} \ot \iden)\right]. \label{eq:g36}
\ea
By the definition of $\ket{\G}$ in \eqref{eq:gg7} and our choice of the state $\ket{A_1}$, the $i = 1$ part of this sum is
\begin{equation*}
\Tr_{\ell k+1, \dots, \ell k + i} \left[(\bra{\varphi^{\ell}} \ot \iden) \cdot \ketbra{\b}{\b} \cdot (\ket{\varphi^{\ell}} \ot \iden)\right]
= \Tr_{1, \dots, i}\ketbra{\G}{\G}.
\end{equation*}
Since $\ket{\b}$ is in $\mps_{(\ell+1) k}(r)$,  $\Tr_{1,\dots,\ell k+i}\ketbra{\b}{\b}$ has rank at most~$r$.
But Equation \eqref{eq:g36} is a sum of PSD operators, so in order to have $\Tr_{\ell k + i}\ketbra{\b}{\b}$ be rank $\leq r$,
the $i = 1$ part of the sum in Equation \eqref{eq:g36} must also be rank~$\leq r$.
This implies that $\Tr_{1, \ldots, i}\ketbra{\G}{\G}$, and therefore also $\Tr_{1, \ldots, i}\ketbra{\widetilde{\G}}{\widetilde{\G}}$, is rank $r$.
As a result, $\ket{\widetilde{\G}}\in \mps_k (r)$, which concludes the proof.
 \end{proof}
 
Applying \propref{tensor-overlap} to $\ket{\Phi_n}$,
we can compute its overlap as
\begin{equation*}
\overlap_r(\ket{\Phi_n})
= \left((1-\theta) + (r-1) \cdot \frac{\theta}{d-1}\right)^{n/2}
\leq \left(1-\frac{\theta}{2}\right)^{n/2},
\end{equation*}
where the first inequality uses $d-1 \geq 2\cdot (r-1)$.
Now if we pick $\theta$ to be
\begin{equation}\label{eq:pick-theta}
\theta =  \frac{8\delta^2}{n},
\end{equation}
we get
\begin{equation*}
\left(1-\frac{4\delta^2}{n}\right)^{n/2}
\leq 1 -\delta^2,
\end{equation*}
for $\delta\leq \tfrac{1}{\sqrt{2}}$, where we have used the inequality $(1-x)^n \leq 1 - \tfrac{1}{2} x n$ for $x \leq \tfrac{1}{n}$.
As a result, the distance of $\ket{\Phi_n}$ to $\mps(r)$ is
\begin{equation*}
\distance_r(\ket{\Phi_n})
= \sqrt{1 - \overlap_r(\ket{\Phi_n})}
\geq \delta.
\end{equation*}
Therefore, $\ket{\Phi_n}$ is far from $\mps(r)$,
and any $\mps(r)$ testing algorithm should reject it with probability at least $\tfrac{2}{3}$.
(Note that because $0 \leq \theta \leq 1$,
we must have $8 \delta^2/n \leq 1$, which is satisfied if $n \geq 4$.)

Our hard family of states which are far from $\mps(r)$
will consist of $\ket{\Phi_n}$ and any state which can be computed from $\ket{\Phi_n}$ by a local unitary.
To make this formal, consider the ensemble of pure states in which a random element is sampled as follows:
first, sample $\boldsymbol{U}_1, \ldots, \boldsymbol{U}_n, \boldsymbol{V}_1, \ldots, \boldsymbol{V}_n \sim \GL{d}$,
i.e.\ $2n$ Haar random $d \times d$ unitary matrices,
and output
\begin{equation*}
(\boldsymbol{U}_1\ot \boldsymbol{V}_1) \ot \cdots \ot (\boldsymbol{U}_n\ot \boldsymbol{V}_n) \cdot \ket{\Phi_n}.
\end{equation*}
Local unitaries do not affect the distance to $\mps(r)$, and so each state in this ensemble is distance~$\delta$ from $\mps(r)$.
Thus, if a tester is given~$m$ copies of any of these states, it should reject with probability at least~$\tfrac{2}{3}$.
As a result, it should also reject with probability at least~$\tfrac{2}{3}$ if given the density matrix
\begin{equation*}
\rho_{\mathrm{far}}
= \E \left((\boldsymbol{U}_1\ot \boldsymbol{V}_1) \ot \cdots \ot (\boldsymbol{U}_n\ot \boldsymbol{V}_n) \cdot \ket{\Phi_n}\bra{\Phi_n}
			\cdot  (\boldsymbol{U}_1^\dagger \ot \boldsymbol{V}_1^\dagger) \ot \cdots \ot (\boldsymbol{U}_n^\dagger \ot \boldsymbol{V}_n^\dagger)\right)^{\otimes m}
\end{equation*}
corresponding to~$m$ copies of a random state drawn from this ensemble.
We will show that this is difficult for an $\mps(r)$ tester unless $m$ is sufficiently large.
To do this, we will show that there exists another density matrix $\rho_{\mathrm{MPS}}$
corresponding to a mixture over states in $\mps(r)$ such that the trace distance
between $\rho_{\mathrm{MPS}}$ and $\rho_{\mathrm{far}}$ is small
unless $m$ is sufficiently large.
To define $\rho_{\mathrm{MPS}}$,
let us first define the state
\begin{equation*}
    \ket{\g}=\sqrt{1-\theta} \cdot \ket{1}\ket{1}+\sum_{i=2}^{r} \sqrt{\frac{\t}{r-1}} \cdot \ket{i}\ket{i},
\end{equation*}
and the state $\ket{\Gamma_n} = \ket{\g}^{\otimes n/2}$.
The state $\ket{\gamma}$ is an element of $\mps(r)$,
and therefore so is $\ket{\Gamma_n}$.
Then we define
\begin{equation*}
\rho_{\mathrm{MPS}}
= \E \left((\boldsymbol{U}_1\ot \boldsymbol{V}_1) \ot \cdots \ot (\boldsymbol{U}_n\ot \boldsymbol{V}_n) \cdot \ket{\Gamma_n}\bra{\Gamma_n}
			\cdot  (\boldsymbol{U}_1^\dagger \ot \boldsymbol{V}_1^\dagger) \ot \cdots \ot (\boldsymbol{U}_n^\dagger \ot \boldsymbol{V}_n^\dagger)\right)^{\otimes m}.
\end{equation*}
Each state in this ensemble is in $\mps(r)$,
and so if a tester is given this density matrix, it should accept with probability at least~$\tfrac{2}{3}$.
Our main result is as follows.

\begin{thm}(Lower bound on copy complexity of MPS testing)\label{thm:mps lower bound probability of success}
Suppose there is an algorithm that accepts $\rho_{\mathrm{MPS}}$ with probability at least $\tfrac{2}{3}$
and accepts $\rho_{\mathrm{far}}$ with probability at most $\tfrac{1}{3}$.
Then $m = \Omega(\sqrt{n}/\delta^2)$.

As a result, $\Omega(\sqrt{n}/\delta^2)$ copies are necessary to test whether a state is in $\mps(r)$ for $\delta\leq\tfrac{1}{\sqrt{2}}$.
\end{thm}

\begin{proof}
Our goal is to bound $\mathrm{D}_{\mathrm{tr}}(\rho_{\mathrm{far}}, \rho_{\mathrm{MPS}})$.
To do so, it is convenient to also work with the fidelity of these states.
Recall that the fidelity $F(\a,\b)$ of two mixed states $\a,\b$ is defined by $F(\a,\b)=\norm{\sqrt{\a}\sqrt{\b}}_1$.
One useful property of this measure is that it is multiplicative with respect to tensor products, i.e.\ $F(\a_1\ot \a_2,\b_1\ot \b_2)=F(\a_1, \b_1)F(\a_2,\b_2)$.
Another is the bound
 \begin{align}
1-F(\a,\b)\leq \mathrm{D}_{\mathrm{tr}}(\a,\b)\leq \sqrt{1-F(\a,\b)^2}\label{eq:gg11}
 \end{align}
between the trace distance $\mathrm{D}_{\mathrm{tr}}(\a,\b)$ and the fidelity $F(\a,\b)$, which we can use to switch back and forth between these two measures.

We begin by applying the upper-bound in~\eqref{eq:gg11} to switch to fidelity:
\begin{equation*}
\mathrm{D}_{\mathrm{tr}}(\rho_{\mathrm{far}}, \rho_{\mathrm{MPS}})\leq \sqrt{1-F(\rho_{\mathrm{far}}, \rho_{\mathrm{MPS}})^2}.
\end{equation*}
Hence, to upper-bound the trace distance between two states, it is sufficient to lower-bound their fidelity.
We note that since $\ket{\Phi_n} = \ket{\varphi}^{\otimes n/2}$, we can rewrite the state $\rho_{\mathrm{far}}$ as
\begin{equation*}
\rho_{\mathrm{far}} = \left(\E_{\boldsymbol{U}, \boldsymbol{V} \sim \GL{d}}
	(\boldsymbol{U} \otimes \boldsymbol{V} \cdot \ket{\varphi}\bra{\varphi}
		\cdot \boldsymbol{U}^\dagger \otimes \boldsymbol{V}^\dagger)^{\otimes m}  \right)^{\otimes n/2}
		=:\sigma_{\mathrm{far}}^{\otimes n/2}.
\end{equation*}
By similar reasoning, we can rewrite $\rho_{\mathrm{MPS}}$ as
\begin{equation*}
\rho_{\mathrm{MPS}} = \left(\E_{\boldsymbol{U}, \boldsymbol{V} \sim \GL{d}}
	(\boldsymbol{U} \otimes \boldsymbol{V} \cdot \ket{\gamma}\bra{\gamma}
		\cdot \boldsymbol{U}^\dagger \otimes \boldsymbol{V}^\dagger)^{\otimes m}  \right)^{\otimes n/2}
		=:\sigma_{\mathrm{MPS}}^{\otimes n/2}
\end{equation*}
Hence, by the multiplicativity of fidelity, we have
\begin{equation*}
F(\rho_{\mathrm{far}}, \rho_{\mathrm{MPS}})
= F(\sigma_{\mathrm{far}}, \sigma_{\mathrm{MPS}})^{n/2}.
\end{equation*}
Now, by applying~\eqref{eq:gg11} again to switch back to trace distance, we have
\begin{equation*}
F(\sigma_{\mathrm{far}}, \sigma_{\mathrm{MPS}})
\geq 1 - \mathrm{D}_{\mathrm{tr}}(\sigma_{\mathrm{far}}, \sigma_{\mathrm{MPS}}).
\end{equation*}
As a result, we would like to upper-bound the trace distance of $\sigma_{\mathrm{far}}$ and $\sigma_{\mathrm{MPS}}$.

Consider an algorithm trying to distinguish these two states.
For $i \in \{1, \ldots, d\}$, let $\ket{\boldsymbol{a}_i} = \boldsymbol{U}\ket{i}$
and let $\ket{\boldsymbol{b}_i} = \boldsymbol{V}\ket{i}$.
Then when the algorithm is given $\sigma_{\mathrm{far}}$,
we can equivalently view it as the algorithm being given $m$ copies of the random sample
\begin{equation*}
\sqrt{1-\theta} \cdot \ket{\boldsymbol{a}_1}\ket{\boldsymbol{b}_1}+\sum_{i=2}^{d} \sqrt{\frac{\t}{d-1}} \cdot \ket{\boldsymbol{a}_i}\ket{\boldsymbol{b}_i},
\end{equation*}
and when it is given $\sigma_{\mathrm{MPS}}$,
we can equivalently view it as being given $m$ copies of the random sample
\begin{equation*}
\sqrt{1-\theta} \cdot \ket{\boldsymbol{a}_1}\ket{\boldsymbol{b}_1}+\sum_{i=2}^{r} \sqrt{\frac{\t}{r-1}} \cdot \ket{\boldsymbol{a}_i}\ket{\boldsymbol{b}_i}.
\end{equation*}
The only difference between these two mixtures
is whether the state has Schmidt coefficients $1-\theta, \theta/(d-1), \ldots, \theta/(d-1)$
or Schmidt coefficients $1-\theta, \theta/(r-1), \ldots, \theta/(r-1)$.
As we show in~\thmref{reduction to one subsystem}, this means that the algorithm learns everything it needs to learn about which case it is in
simply by measuring the $m$ $\ket{\boldsymbol{a}_i}$ registers,
and it can ignore the $m$ $\ket{\boldsymbol{b}_i}$ registers.
In other words, if we set
\begin{equation*}
\tau_{\mathrm{far}} = \Tr_2 \ket{\varphi}\bra{\varphi} = (1-\theta) \cdot \ket{1}\bra{1} + \theta \cdot \sum_{i=2}^d \frac{1}{d-1}\cdot \ket{i}\bra{i}
\end{equation*}
and
\begin{equation*}
\tau_{\mathrm{MPS}} = \Tr_2 \ket{\gamma}\bra{\gamma} = (1-\theta) \cdot \ket{1}\bra{1} + \theta \cdot \sum_{i=2}^r \frac{1}{r-1}\cdot\ket{i}\bra{i},
\end{equation*}
then
\begin{equation*}
\mathrm{D}_{\mathrm{tr}}(\sigma_{\mathrm{far}}, \sigma_{\mathrm{MPS}})
= \mathrm{D}_{\mathrm{tr}}\left(\E_{\boldsymbol{U} \sim \GL{d}}(\boldsymbol{U}\tau_{\mathrm{far}} \boldsymbol{U^{\dagger}})^{\ot m},
		\E_{\boldsymbol{U} \sim \GL{d}}(\boldsymbol{U}\tau_{\mathrm{MPS}} \boldsymbol{U^{\dagger}})^{\ot m} \right).
\end{equation*}

The density matrix $\E_{\boldsymbol{U}\sim \GL{d}}(\boldsymbol{U}\tau_{\mathrm{far}}\boldsymbol{U^{\dagger}})^{\ot m}$
can be described by the following mixture.
Let $\ket{\boldsymbol{a}_1},\dots,\ket{\boldsymbol{a}_d}$ be a random orthonormal basis for $\bbC^d$ as above. Draw $m$ samples as follows.
\begin{itemize}
	\item[(i)] With probability $1-\t$, output $\ket{\boldsymbol{a}_1}$.
	\item[(ii)] With probability $\t$, output one of the states $\ket{\boldsymbol{a}_2},\dots,\ket{\boldsymbol{a}_d}$ uniformly at random. 
\end{itemize}
The state $\E_{\boldsymbol{U}\sim \GL{d}}(\boldsymbol{U}\tau_{\mathrm{MPS}}\boldsymbol{U^{\dagger}})^{\ot m}$ can be described by a similar mixture
except that now in step (ii), with probability $\t$, the output is one of the states $\ket{\boldsymbol{a}_2},\dots,\ket{\boldsymbol{a}_r}$ chosen uniformly at random. Consider the event that either all the $m$ draws are from step (i) or $m-1$ draws are from step (i) and the remaining sample is from step (ii). The probability of this event occurring is simply $(1-\t)^m+m\ \t(1-\t)^{m-1}$.  In both of these cases, it is not possible to distinguish the two states. In all the other cases, where more than one sample is drawn according to step (ii), we loosely upper bound the distance between the states by $1$. This gives us the following overall upper bound on the distance between the random ensembles:
\begin{align}
	\mathrm{D}_{\mathrm{tr}}\left(\E_{\boldsymbol{U} \sim \GL{d}}(\boldsymbol{U}\tau_{\mathrm{far}} \boldsymbol{U^{\dagger}})^{\ot m},
			\E_{\boldsymbol{U} \sim \GL{d}}(\boldsymbol{U}\tau_{\mathrm{MPS}} \boldsymbol{U^{\dagger}})^{\ot m}\right)&\leq 1-(1-\t)^m-m\t(1-\t)^{m-1}\nn\\
	&\leq 1-(1-m \t)-m\t\left(1-(m-1)\t\right)\nn\\
	&=m(m-1)\t^2\nn.
\end{align}
As a result, this implies that $\rho_{\mathrm{far}}$ and $\rho_{\mathrm{MPS}}$ have distance
\begin{align*}
\mathrm{D}_{\mathrm{tr}}(\rho_{\mathrm{far}}, \rho_{\mathrm{MPS}})
& \leq \left(1 - \left(1 -m(m-1)\t^2\right)^n\right)^{1/2}\\
& \leq \left(1 - \left(1 - n \cdot m(m-1)\t^2\right)\right)^{1/2}\\
& \leq \sqrt{n} m \theta.
\end{align*}
By our choice of $\theta = 8\delta^2/n$ in Equation~\eqref{eq:pick-theta},
this is at most $4 m \delta^2/\sqrt{n}$.
For an algorithm to accept $\rho_{\mathrm{MPS}}$ with probability at least $\tfrac{2}{3}$
and $\rho_{\mathrm{far}}$ with probability at most $\tfrac{1}{3}$,
this trace distance must be at least $\tfrac{1}{3}$.
This implies that $m$ must be at least $\tfrac{1}{24}\sqrt{n}/\delta^2$,
which completes the proof.
\end{proof}
Here we prove the claim in the proof of \thmref{mps lower bound probability of success} that it suffices for any algorithm that tries to distinguish between the states $\s_{\mathrm{MPS}}$ and $\s_{\mathrm{far}}$ to only measure their $m$ $\ket{\boldsymbol{a}_i}$ registers.
The proof is standard and based on repeated applications of Schur's Lemma.
\begin{thm}\label{thm:reduction to one subsystem}
Let $\ket{\varphi}_{AB}$ and $\ket{\g}_{AB}$ be two bipartite states on subsystems $A$ and $B$. Any algorithm for distinguishing between the two mixed states
\begin{align}
	\E_{\boldsymbol{U}_A\sim \GL{d},\boldsymbol{V}_B\sim \GL{d}}\left(\boldsymbol{U}_A\otimes \boldsymbol{V}_B\cdot \ketbra{\varphi}{\varphi} \cdot \boldsymbol{U}_A^{\dagger}\otimes\boldsymbol{V}_B^{\dagger}\right)^{\ot m}
\end{align}
and 
\begin{align}
	\E_{\boldsymbol{U}_A\sim \GL{d},\boldsymbol{V}_B\sim \GL{d}}\left(\boldsymbol{U}_A\otimes \boldsymbol{V}_B\cdot \ketbra{\g}{\g} \cdot \boldsymbol{U}_A^{\dagger}\otimes\boldsymbol{V}_B^{\dagger}\right)^{\ot m}
\end{align}
can without loss of generality leave out the $m$ $B$ registers and only measure the $m$ $A$ registers. 
\end{thm}
We denote the Schur-Weyl basis for $(\bbC^d)^{\ot m}$ (see \thmref{Schur Weyl})  by $\ket{\m}\ket{q}\ket{p}$,
where $q$ is a basis vector for the $\m$-irrep $\GLirrep_{\m}^d$ of $\GL{d}$ and $p$ is a basis vector for the $\m$-irrep $\srep_{\mu}$ of $\symm{d}$.
One technicality is that although any orthonormal basis $\{\ket{q}\}_q$ of $\GLirrep_{\m}^d$ will suffice for our purposes,
we will need to pick a basis $\{\ket{p}\}_p$ of $\srep_{\mu}$ such that the matrix entries of $\srep_{\mu}(\pi)$ are real-valued for each $\pi \in \symm{d}$.
(This is used to establish Equation~\eqref{eq:here's-where-we-used-it} below.)
One basis that satisfies this property is known as the \emph{Gelfand-Tsetlin basis},
and the resulting matrices $\{\srep_{\mu}(\pi)\}_{\pi \in \symm{d}}$ give rise to \emph{Young's orthogonal representation}.
In this basis, the matrix elements $\srep_{\mu}(\pi)_{p, p'}:= \bra{p} \srep_{\mu}(\pi)\ket{p'}$ are real-valued,
and so each matrix $\srep_{\mu}(\pi)$ is an orthogonal matrix.
For an introduction to the Gelfand-Tsetlin basis, see~\cite[Appendix~B]{HGG09} and the citations contained therein.

Before proving \thmref{reduction to one subsystem}, we show some helper lemmas.

\begin{lem} \label{lem:rotate}
Let $\mathcal{H} = (\bbC^d)^{\otimes m}$,
and let $\mathcal{H}'$ be another Hilbert space.
Consider a matrix $N$ acting on $\mathcal{H} \otimes \mathcal{H'}$ of the form  
 \begin{align}
 	N=\sum_{\substack{\m,\m'\\ q,q'}} \ketbra{\m}{\m'} \ot \ketbra{q}{q'} \ot N_{\m,\m',q,q'},\label{eq:gg18}
 \end{align}
 where $N_{\m,\m',q,q'}$ is an operator acting on $\mathcal{H'}$ and  the $\ket{p}$ register of $\mathcal{H}$. Then it holds that 
\begin{align}
	\E_{\boldsymbol{U}\sim \GL{d}}\left(\boldsymbol{U}^{\ot m} \cdot N\cdot (\boldsymbol{U^{\dagger}})^{\ot m}\right)&=\sum_{\m} \ketbra{\m}{\m} \ot \iden \ot N_{\m},\label{eq:gg19}
\end{align}
where $N_{\m}$ is an operator acting on $\mathcal{H'}$ and  the $\ket{p}$ register of $\mathcal{H}$.
\end{lem}
\begin{proof}
To begin, we calculate
\begin{equation}\label{eq:thing-im-calculating}
\E_{\boldsymbol{U}\sim \GL{d}}\left(\boldsymbol{U}^{\ot m} \cdot N\cdot (\boldsymbol{U^{\dagger}})^{\ot m}\right)
= \sum_{\substack{\m,\m'\\ q,q'}} \ketbra{\m}{\m'} \ot 
	\left( \E_{\boldsymbol{U}\sim \GL{d}} \GLirrep_\m^d(\boldsymbol{U})\cdot \ketbra{q}{q'}\cdot \GLirrep_{\m'}^d(\boldsymbol{U})^\dagger \right) \ot N_{\m,\m',q,q'}.
\end{equation}
For each $\mu, \mu', q, q'$, the matrix
\begin{equation*}
T_{\mu, \mu', q, q'} = \E_{\boldsymbol{U}\sim \GL{d}}[ \GLirrep_\m^d(\boldsymbol{U})\cdot \ketbra{q}{q'}\cdot \GLirrep_{\m'}^d(\boldsymbol{U})^\dagger ]
\end{equation*} 
is an intertwining operator operator for $\GLirrep_\m^d$ and $\GLirrep_{\m'}^d$, because for each $V \in \GL{d}$,
\begin{align*}
\GLirrep_\m^d(V) \cdot T_{\mu, \mu', q, q'}
&= \GLirrep_\m^d(V) \cdot\E_{\boldsymbol{U}\sim \GL{d}}[ \GLirrep_\m^d(\boldsymbol{U})\cdot \ketbra{q}{q'} \cdot\GLirrep_{\m'}^d(\boldsymbol{U})^\dagger ]\\
&= \E_{\boldsymbol{U}\sim \GL{d}}[ \GLirrep_\m^d(V\boldsymbol{U})\cdot \ketbra{q}{q'}\cdot \GLirrep_{\m'}^d(\boldsymbol{U})^\dagger] \\
&= \E_{\boldsymbol{W}\sim \GL{d}}[ \GLirrep_\m^d(\boldsymbol{W})\cdot \ketbra{q}{q'} \cdot \GLirrep_{\m'}^d(V^\dagger \boldsymbol{W})^\dagger] \\
&= \E_{\boldsymbol{W}\sim \GL{d}}[ \GLirrep_\m^d(\boldsymbol{W})\cdot \ketbra{q}{q'} \cdot \GLirrep_{\m'}^d(\boldsymbol{W})^\dagger] \cdot \GLirrep_{\m'}^d(V)\\
&=T_{\mu, \mu', q, q'} \cdot \GLirrep_{\m'}^d(V).
\end{align*}
As a result, Schur's lemma states that $T_{\mu, \mu', q, q'}$ is zero when $\mu \neq \mu'$,
and a multiple of the identity $c_{\mu, q, q'} \cdot \iden$ when $\mu = \mu'$.
Indeed, we may compute $c_{\mu, q, q'}$ exactly as
\begin{align*}
c_{\mu, q, q'} 
= \frac{1}{\dim(\GLirrep_\m^d)} \cdot \Tr[T_{\mu, \mu, q, q'}]
&= \frac{1}{\dim(\GLirrep_\m^d)} \cdot \E_{\boldsymbol{U}\sim \GL{d}}[ \Tr[\GLirrep_\m^d(\boldsymbol{U})\cdot \ketbra{q}{q'}\cdot \GLirrep_{\m}^d(\boldsymbol{U})^\dagger] ]\\
&= \frac{1}{\dim(\GLirrep_\m^d)} \cdot \E_{\boldsymbol{U}\sim \GL{d}}[\bra{q'}\cdot \GLirrep_{\m}^d(\boldsymbol{U})^\dagger\GLirrep_\m^d(\boldsymbol{U})\cdot \ket{q}]\\
&= \frac{1}{\dim(\GLirrep_\m^d)} \cdot \E_{\boldsymbol{U}\sim \GL{d}}[\braket{q'}{q}]
 = \left\{\begin{array}{cl}
1/ \dim(\GLirrep_\m^d) & \text{if $q = q'$,}\\
0 & \text{otherwise}.
\end{array}\right.
\end{align*}
Overall, then, $T_{\mu, \mu', q, q'}$ is $(1/ \dim(\GLirrep_\m^d)) \cdot \iden$ if $\mu = \mu'$ and $q = q'$ and zero otherwise.
Thus,
\begin{equation*}
\eqref{eq:thing-im-calculating}
= \sum_{{\m, q}} \ketbra{\m}{\m} \ot \Big(\frac{1}{\dim(\GLirrep_\m^d)}\cdot \iden\Big) \ot N_{\m,\m,q,q}
= \sum_{\m} \ketbra{\m}{\m}\ot \iden \ot \Big(\frac{1}{\dim(\GLirrep_\m^d)}\cdot \sum_{q} N_{\mu, \mu, q, q}\Big).
\end{equation*}
The lemma follows by taking $N_\mu =(1/ \dim(\GLirrep_\m^d)) \cdot  \sum_{q} N_{\mu, \mu, q, q}$.
\end{proof}

\begin{lem}[EPR state in an irrep]\label{lem:EPR}
Given $\mu \vdash m$, we define the EPR state corresponding to the permutation irrep $\sirrep_\mu$ as
\begin{equation*}
\ket{\mathrm{EPR}_\mu} = \frac{1}{\sqrt{\dim(\mu)}} \cdot \sum_p \ket{p} \otimes \ket{p},
\end{equation*}
where the sum ranges over basis vectors of $\sirrep_\mu$.
Then
\begin{equation*}
\E_{\boldsymbol{\pi}\sim \symm{d}}[\srep_A(\boldsymbol{\pi}) \ot  \srep_B(\boldsymbol{\pi})]
= \sum_{\mu} \ketbra{\mu}{\mu}_A\otimes \ketbra{\mu}{\mu}_B \otimes \iden_A \otimes \iden_B
		\otimes \ketbra{\mathrm{EPR}_{\mu}}{\mathrm{EPR}_{\mu}}_{A, B},
\end{equation*}
where the two identity matrices act on the $\ket{q}$ registers of Hilbert spaces $A$ and $B$.
\end{lem}
\begin{proof}
We begin by calculating
\begin{equation}\label{eq:about-to-use-orthogonality}
\E_{\boldsymbol{\pi}\sim \symm{d}}[\sirrep_{\m_A}(\boldsymbol{\pi}) \ot  \sirrep_{\m_B}(\boldsymbol{\pi})]
 = \sum_{\substack{p_A, p_A' \\ p_B, p_B'}} \ketbra{p_A}{p_A'}\ot\ketbra{p_B}{p_B'} \cdot
	\E_{\boldsymbol{\pi}\sim \symm{d}}[\sirrep_{\m_A}(\boldsymbol{\pi})_{p_A, p_A'} \cdot  \sirrep_{\m_B}(\boldsymbol{\pi})_{p_B, p_B'}].
\end{equation}
The Schur orthogonality relations state that
\begin{equation*}
\E_{\boldsymbol{\pi}\sim \symm{d}}[\sirrep_{\m_A}(\boldsymbol{\pi})_{p_A, p_A'}^\dagger \cdot  \sirrep_{\m_B}(\boldsymbol{\pi})_{p_B, p_B'}]
= \left\{\begin{array}{cl}
1/ \dim(\mu_A) & \text{if $\m_A = \m_B$, $p_A = p_B$, and $p_A' = p_B'$,}\\
0 & \text{otherwise}.
\end{array}\right.
\end{equation*}
Recall that we have chosen our basis of $\sirrep_\mu$ so that $\sirrep_\mu(\pi)$ is a real-valued (orthogonal) matrix for each $\pi \in \symm{d}$.
Then
\begin{equation*}
\sirrep_{\m_A}(\pi)_{p_A, p_A'}^\dagger = \sirrep_{\m_A}(\pi)_{p_A, p_A'},
\end{equation*}
and so
\begin{equation}\label{eq:here's-where-we-used-it}
\E_{\boldsymbol{\pi}\sim \symm{d}}[\sirrep_{\m_A}(\boldsymbol{\pi})_{p_A, p_A'} \cdot  \sirrep_{\m_B}(\boldsymbol{\pi})_{p_B, p_B'}]
= \left\{\begin{array}{cl}
1/ \dim(\mu_A) & \text{if $\m_A = \m_B$, $p_A = p_B$, and $p_A' = p_B'$,}\\
0 & \text{otherwise}.
\end{array}\right.
\end{equation}
As a result, \eqref{eq:about-to-use-orthogonality} is zero if $\m_A \neq \mu_B$,
and
\begin{equation*}
\eqref{eq:about-to-use-orthogonality}
= \frac{1}{\dim(\mu_A)} \cdot \sum_{p, p'} \ketbra{p}{p'} \otimes \ketbra{p}{p'}
= \ketbra{\mathrm{EPR}_{\mu_A}}{\mathrm{EPR}_{\mu_A}}.
\end{equation*}
if $\mu_A = \mu_B$. This allows us to express
\begin{align*}
\E_{\boldsymbol{\pi}\sim \symm{d}}[\srep_A(\boldsymbol{\pi}) \ot  \srep_B(\boldsymbol{\pi})]
&= \sum_{\mu_A, \mu_B} \ketbra{\mu_A}{\mu_A}\otimes \ketbra{\mu_B}{\mu_B} \otimes \iden \otimes \iden
		\otimes \E_{\boldsymbol{\pi}\sim \symm{d}}[\sirrep_{\m_A}(\boldsymbol{\pi}) \ot  \sirrep_{\m_B}(\boldsymbol{\pi})]\\
&= \sum_{\mu} \ketbra{\mu}{\mu}\otimes \ketbra{\mu}{\mu} \otimes \iden \otimes \iden
		\otimes \ketbra{\mathrm{EPR}_{\mu}}{\mathrm{EPR}_{\mu}}.
\end{align*}
This completes the proof.
\end{proof}

Next, we have the following immediate corollary of \lemref{EPR}.

\begin{cor}\label{cor:apply-epr}
Consider an operator of the form
\begin{equation*}
O = \sum_{\mu_A, \mu_B} \ketbra{\mu_A}{\mu_A}\ot \ketbra{\mu_B}{\mu_B} \ot \iden_A \ot \iden_B \ot O_{\mu_A, \mu_B},
\end{equation*}
where the two identity matrices act on the $\ket{q}$ registers of Hilbert spaces $A$ and $B$,
and the $O_{\mu_A, \mu_B}$ matrix acts on the $\ket{p}$ registers of $A$ and $B$,
Next, let $Z$ be the matrix
\begin{equation*}
Z = \E_{\boldsymbol{\pi}\sim \symm{d}}[\srep_A(\boldsymbol{\pi}) \ot  \srep_B(\boldsymbol{\pi})].
\end{equation*}
Then
\begin{equation*}
Z \cdot O \cdot Z = \sum_\mu c_\mu \cdot \ketbra{\mu}{\mu}_A \ot \ketbra{\mu}{\mu}_B \ot \iden_A \ot \iden_B \ot \ketbra{\mathrm{EPR}_{\mu}}{\mathrm{EPR}_{\mu}},
\end{equation*}
for some constants $c_\mu$.
\end{cor}
\begin{proof}
By \lemref{EPR},
\begin{align*}
Z \cdot O \cdot Z
& = \sum_{\mu} \ketbra{\mu}{\mu}_A\ot \ketbra{\mu}{\mu}_B \ot \iden_A \ot \iden_B
	\ot (\ketbra{\mathrm{EPR}_{\mu}}{\mathrm{EPR}_{\mu}} \cdot O_{\mu, \mu} \cdot \ketbra{\mathrm{EPR}_{\mu}}{\mathrm{EPR}_{\mu}})\\
& = \sum_\mu c_\mu \cdot \ketbra{\mu}{\mu}_A \ot \ketbra{\mu}{\mu}_B \ot \iden_A \ot \iden_B \ot \ketbra{\mathrm{EPR}_{\mu}}{\mathrm{EPR}_{\mu}},
\end{align*}
where $c_\mu = \bra{\mathrm{EPR}_\mu}\cdot O_{\mu, \mu}\cdot \ket{\mathrm{EPR}_\mu}$. This completes the proof.
\end{proof}

Now we prove \thmref{reduction to one subsystem}.

\begin{proof}[Proof of \thmref{reduction to one subsystem}]
Given $\psi \in \{\varphi, \gamma\}$, consider the state $M_\psi$ defined as
	\[M_\psi:=\E_{\boldsymbol{U}_A\sim \GL{d},\boldsymbol{V}_B\sim \GL{d}}\left(\boldsymbol{U}_A\otimes \boldsymbol{V}_B\cdot \ketbra{\psi}{\psi} \cdot \boldsymbol{U}_A^{\dagger}\otimes\boldsymbol{V}_B^{\dagger}\right)^{\ot m}.\]
Using the left and right invariance property of the Haar measure and the commutation between $\srep_A$ and $\GLrep_A$ (and likewise for $\srep_B$ and $\GLrep_B$), we can see that the mixed state $M_\psi$ remains invariant under the following permutations and unitary rotations:

\begin{multicols}{2}
\begin{enumerate}
\item $\E_{\boldsymbol{U}\sim \GL{d}}\left(\boldsymbol{U}_A^{\ot m} \cdot M_\psi\cdot (\boldsymbol{U}_A^{\dagger})^{\ot m}\right)=M_\psi$,
\item $\E_{\boldsymbol{V}\sim \GL{d}}\left(\boldsymbol{V}_B^{\ot m}\cdot M_\psi \cdot (\boldsymbol{V}_B^{\dagger})^{\ot m}\right)=M_\psi$,
\item $\E_{\boldsymbol{\pi}\sim \symm{d}}\left(\srep_A(\boldsymbol{\pi}) \ot  \srep_B(\boldsymbol{\pi})\cdot M_\psi \right)=M_\psi$,
\item $\E_{\boldsymbol{\pi}\sim \symm{d}}\left(M_\psi\cdot \srep_A(\boldsymbol{\pi}) \ot  \srep_B(\boldsymbol{\pi})\right)=M_\psi$.
\end{enumerate}
\end{multicols}
We can now apply the results of \lemref{rotate} and \corref{apply-epr} to put the mixed state $M_\psi$ in the following form
 \begin{align}
M_\psi=\sum_{\m} c_{\psi, \mu} \cdot \ketbra{\m}{\m}_A \ot \ketbra{\m}{\m}_B  \ot \iden_A \ot \iden_B \ot \ketbra{\mathrm{EPR}_\m}{\mathrm{EPR}_\m}_{AB}.
 \end{align}
We can therefore interpret the density matrix $M_\psi$ as corresponding to a mixed state with one element in the mixture for each block $\mu$.
In this case, measuring the block $\mu$ can be done without loss of generality, as it does not perturb the state.
That can be done entirely on the~$A$ subsystem by performing the projective measurement $\{\ketbra{\m}{\mu}_A\}$.
Having done this measurement and received outcome $\boldsymbol{\mu}$, the state is equal to
\begin{equation*}
\ketbra{\boldsymbol{\mu}}{\boldsymbol{\mu}}_A \ot \ketbra{\boldsymbol{\mu}}{\boldsymbol{\mu}}_B  \ot \iden_A \ot \iden_B \ot \ketbra{\mathrm{EPR}_{\boldsymbol{\mu}}}{\mathrm{EPR}_{\boldsymbol{\mu}}}_{AB},
\end{equation*}
 regardless of whether $\psi = \varphi$ or $\gamma$.
 Hence, no further information can be learned about~$\psi$ by performing any further measurements,
 and this implies that measuring only the $A$ subsystem is without loss of generality.
\end{proof}

 \section*{Acknowledgements} We thank Rolando La Placa for useful discussions in the early stage of this work. We also thank Aram Harrow for helpful discussions and for helpful feedback on an earlier draft of this work and Ashley Montanaro for helpful discussions. MS was supported by NSF grant CCF-1729369 and a Samsung Advanced Institute of Technology Global Research Partnership. JW was funded by ARO contract W911NF-17-1-0433.

\bibliographystyle{alpha}
\bibliography{main}

\end{document}